\documentclass[11pt, reqno]{amsart}

\usepackage{amsmath}
\usepackage{amssymb}
\usepackage{amsthm}
\usepackage{cite}
\usepackage{doi}
\usepackage{mathtools}
\usepackage{microtype}

\newtheorem{theorem}{Theorem}
\newtheorem{definition}{Definition}
\newtheorem{remark}{Remark}
\newtheorem{example}{Example}

\newtheorem{corollary}{Corollary}
\newtheorem{proposition}{Proposition}
\newtheorem{lemma}{Lemma}

\newcommand{\gbinom}[3]{\Bigl[ \hspace*{-0.2cm} \begin{array}{c} {\hspace*{0.2cm} #1} \\ 
{\hspace*{-0.3cm} #2} \hspace*{-0.5cm} \end{array}\Bigr]_{#3}}  

\allowdisplaybreaks

\newcommand{\sdpinner}[2]{\left\langle #1,#2\right\rangle}

\begin{document}

\title[Shannon Capacity for Lexicographic Graph Products]
{Shannon Capacity and Related Graph Invariants for Lexicographic Products}

\author[I. Sason]{Igal Sason}

\address{\normalfont \newline 
Igal Sason is with the Viterbi Faculty of Electrical and Computer Engineering and the Faculty
of Mathematics, Technion--Israel Institute of Technology, Haifa 3200003, Israel.
E-mail: \texttt{eeigal@technion.ac.il.}}

\begin{abstract}
This paper studies the Shannon capacity of lexicographic products of finite
simple graphs, together with the Lov\'{a}sz theta function and the fractional
Haemers number. The Shannon capacity is proved to be supermultiplicative under
lexicographic products in either order, and these products are compared with
the strong product. We explicitly construct three countably infinite families
of lexicographic powers based on the Schl\"{a}fli graph, the McLaughlin graph,
and its second subconstituent; in each family, pairing each member with its
complement yields strict supermultiplicativity and arbitrarily large
multiplicative gaps. Bounds and exact-capacity criteria for lexicographic
products are derived, and the resulting upper bounds are shown to be
incomparable. The capacities of lexicographic products involving Kneser
graphs, their complements, and \texorpdfstring{\(q\)}{q}-analogues of Kneser
graphs are determined. It is also shown that a lexicographic product with 
a complete outer factor preserves the Shannon capacity of an arbitrary 
inner factor. The capacities of iterated lexicographic powers
are determined, including those of self-complementary graphs that are
vertex-transitive or strongly regular. Elementary, self-contained proofs are
also given for three known results: the multiplicativity of the Lov\'{a}sz
theta function and the fractional Haemers number under lexicographic products,
and the equality of the fractional and ordinary Lov\'{a}sz theta functions.
Finally, an open problem concerning the Shannon capacities of lexicographic
and strong products is posed.
\end{abstract}

\maketitle
\thispagestyle{empty}
\setcounter{page}{1}

\vspace*{-0.4cm}


\section{Introduction}

The lexicographic product is a fundamental graph operation that arises
naturally by replacing each vertex of one graph with a copy of another.
For finite simple graphs \(G\) and \(H\), their \emph{lexicographic product},
denoted by \(G\circ H\) or \(G[H]\), has vertex set
\(V(G)\times V(H)\), where two distinct vertices \((g,h)\) and
\((g',h')\) are adjacent if and only if
\[
g\sim_G g'
\quad \vee \quad
\bigl(g=g' \; \wedge \; h\sim_H h'\bigr).
\]
Here, the relations \(\sim_G\) and \(\sim_H\) denote adjacency in \(G\)
and \(H\), respectively. Equivalently, each vertex \(g\in V(G)\) is
replaced by a copy of \(H\), and whenever \(g\sim_G g'\), every vertex
in the copy corresponding to \(g\) is joined to every vertex in the
copy corresponding to \(g'\).

The lexicographic product provides a natural setting for studying the
behavior of several important graph invariants, including the
Lov\'{a}sz theta function, the fractional Haemers number, and the
Shannon capacity. It also admits a natural interpretation in zero-error
information theory. To formulate this interpretation, we first recall
the graph-theoretic description of zero-error communication.

Given a discrete memoryless channel with input alphabet \(\mathcal{X}\),
its \emph{confusability graph} \(G\) has vertex set \(\mathcal{X}\), with
two distinct input symbols adjacent if they can produce a common channel
output and therefore cannot be distinguished with zero error probability 
\cite{Shannon56}. Thus, a collection of pairwise nonconfusable symbols
forms an independent set in \(G\). Consequently, the maximum number of
messages that can be transmitted without error in a single use of the
channel is \(\alpha(G)\), the \emph{independence number} of \(G\), defined
as the maximum cardinality of an independent set in \(G\).

For \(n\) channel uses, the resulting confusability graph is the
\emph{\(n\)-fold strong power} \(G^{\boxtimes n}\), whose vertex set is
\(V(G)^n\), with two distinct vertices adjacent if, in each coordinate,
their entries are either identical or adjacent in \(G\). Consequently,
the maximum number of messages that can be transmitted without error in
\(n\) channel uses is \(\alpha(G^{\boxtimes n})\). The \emph{Shannon capacity} 
of \(G\) is defined by
\begin{align}
\Theta(G) &\coloneqq
\sup_{n \geq 1} \sqrt[n]{\alpha\bigl(G^{\boxtimes n}\bigr)} \nonumber\\
&= \lim_{n \to \infty} \sqrt[n]{\alpha\bigl(G^{\boxtimes n}\bigr)},
\end{align}
and the limit exists and equals the supremum by Fekete's lemma, since
\begin{align}
\label{eq3: 06.08.26}
\alpha\bigl(G^{\boxtimes(m+n)}\bigr)
\geq \alpha\bigl(G^{\boxtimes m}\bigr) \; \alpha\bigl(G^{\boxtimes n}\bigr),
\qquad n,m \geq 1.
\end{align}
Thus, \(\Theta(G)\) is the asymptotic per-channel-use growth factor of
the maximum number of messages that can be transmitted with zero error
probability, whereas \(\log\Theta(G)\) is the corresponding asymptotic
exponential growth rate \cite{Shannon56}.

Despite the simplicity of its definition, \(\Theta(G)\) is notoriously
difficult to determine. Its exact value remains unknown even for some
small simple graphs, such as the cycle of length~\(7\), and the
independence numbers of successive strong powers can exhibit remarkably
complicated behavior. Alon and Lubetzky \cite{AlonL06} showed, in
particular, that the Shannon capacity cannot in general be approximated
within a subpolynomial factor in the number of vertices using any fixed,
but arbitrarily long, initial segment of the sequence
\[
\bigl\{\alpha(G^{\boxtimes n})\bigr\}_{n \geq 1} \, .
\]
The subject and its connections with graph powers, extremal
combinatorics, coding theory, and spectral methods are surveyed from
several perspectives in \cite{Alon02,Jurkiewicz14,LaviSason2026}.

A remarkable development was the introduction by Lov\'{a}sz of the theta
function \(\vartheta(G)\), which is computable to arbitrary prescribed
accuracy in polynomial time by semidefinite programming and satisfies
the sandwich inequalities
\begin{align}
\label{eq: sandwich}
\alpha(G) \leq \Theta(G) \leq \vartheta(G) \leq \chi(\overline{G}).
\end{align}
Thus, \(\vartheta(G)\) provides an upper bound on the Shannon capacity
of \(G\) and a lower bound on the chromatic number of its complement
\(\overline{G}\) \cite{Lovasz79}. In particular, by combining the upper
bound \(\vartheta(C_5)=\sqrt{5}\) with the standard construction showing
that \(\alpha(C_5^{\boxtimes 2})=5\), Lov\'{a}sz obtained the celebrated
identity
\[
\Theta(C_5)=\sqrt{5}.
\]

The Lov\'{a}sz theta function also lies at the heart of the sandwich
theorem and the interplay among orthonormal representations and
semidefinite optimization \cite{Alon19,Lovasz79,Lovasz19}, graph
parameters \cite{Alon19,Knuth94,Sason24}, and coding-theoretic bounds
\cite{Dalai_ISIT13a,Dalai13_ISIT13b,Dalai_IT13,DuanSW_IT13}.
Another important upper bound on the Shannon capacity is Haemers'
linear-algebraic minimum-rank bound \cite{Haemers79}, along with its
fractional variant \cite{BukhC19}. Hu, Tamo, and Shayevitz subsequently
developed a linear-programming bound that can improve upon both the
Haemers minimum-rank bound and the Lov\'{a}sz theta bound \cite{HuTS18};
see also \cite{LaviSason2026,Sason23,Sason24} for discussions of
theta-type, spectral, and related graph invariants.

We next describe the information-theoretic interpretation of the
lexicographic product. Let \(G\) and \(H\) be finite confusability
graphs, and suppose that the channel input alphabet is
\(V(G)\times V(H)\). Here, \(g\in V(G)\) serves as a
\emph{class label}, whereas \(h\in V(H)\) is a
\emph{within-class symbol}. For each \(g\in V(G)\), let
\[
H_g\coloneqq\bigl\{(g,h):h\in V(H)\bigr\}.
\]
Assume that the confusability graph induced on each class \(H_g\) is a
copy of \(H\), and that the confusability relations between distinct
classes depend only on their class labels. More precisely, for
\(g\neq g'\), every input in \(H_g\) is confusable with every input in
\(H_{g'}\) if \(g\sim_G g'\), whereas no input in \(H_g\) is confusable
with any input in \(H_{g'}\) if \(g\not\sim_G g'\). It follows that
\begin{align}
(g,h)\sim(g',h')
\quad\Longleftrightarrow\quad
(g\sim_G g')
\;\vee\;
\bigl(g=g'\,\wedge\,h\sim_H h'\bigr).
\end{align}
Thus, the confusability graph of the channel is exactly \(G\circ H\).

The uniformity of the confusability relations between distinct classes
is an explicit assumption of this model, rather than a consequence of
representing the channel inputs as ordered pairs. If within-class
symbols also affected confusability between distinct classes, the
resulting graph would not, in general, be a lexicographic product.
Accordingly, \(\Theta(G\circ H)\) is the zero-error capacity of this
class-structured channel.

\begin{example}
{\em
Let \(G=\mathrm{KG}(n,k)\) and \(H=\mathrm{KG}(m,\ell)\) be Kneser graphs, 
where \(n\geq 2k\) and \(m\geq 2\ell\); for background on these graphs, see
Section~\ref{subsection: Kneser graphs}. The channel inputs are pairs
\begin{align}
(A,B)\in
\binom{[n]}{k}\times\binom{[m]}{\ell}.
\end{align}
Two distinct inputs \((A,B)\) and \((A',B')\) are confusable if and only
if
\begin{align}
\bigl(A\cap A'=\varnothing\bigr)
\;\vee\;
\bigl(A=A'\,\wedge\,B\cap B'=\varnothing\bigr).
\end{align}
Hence, the confusability graph is the lexicographic product
\[
\mathrm{KG}(n,k)\circ\mathrm{KG}(m,\ell),
\]
whose Shannon capacity is determined in
Theorem~\ref{theorem: lexicographic products of Kneser-type graphs}.}
\end{example}

Lexicographic products also arise naturally in index coding and network
coding, where they provide a structured means of constructing larger
coding problems from smaller ones. Blasiak, Kleinberg, and Lubetzky
represented index-coding instances by directed hypergraphs and defined
a lexicographic product for such instances. By taking repeated
lexicographic products, they constructed families exhibiting
polynomially growing gaps between various combinatorial bounds and
achievable coding rates, including a polynomial separation between
linear and nonlinear network-coding rates
\cite{BlasiakKleinbergLubetzky2011}. Arbabjolfaei and Kim subsequently
proved that the broadcast rate is multiplicative under the lexicographic
product of side-information graphs \cite{ArbabjolfaeiKim2015}. They later
introduced a generalized lexicographic product for directed
side-information graphs and characterized the capacity region of the
composite index-coding problem in terms of the capacity regions of its
constituent subproblems \cite{ArbabjolfaeiKim2020}; see also the monograph
\cite{ArbabjolfaeiKim2018} for a comprehensive treatment of index coding.
Whereas these applications concern the construction and decomposition
of coding problems, the present work studies the Shannon capacity of
undirected lexicographic products, with the factor graphs interpreted as
confusability graphs.

The purpose of this paper is to study the Shannon capacity and related
graph invariants for lexicographic products. Our results fall into two
categories. First, we present alternative, elementary, and self-contained
proofs of several known results concerning the Lov\'{a}sz theta function
and the fractional Haemers number. Second, we derive new bounds and establish 
strict supermultiplicativity results for lexicographic products, determine 
the exact Shannon capacity for several families of such products, and pose 
an open problem concerning the relationship between the Shannon capacities 
of lexicographic and strong products.

We begin by recording structural properties of the lexicographic product
that will be used in the subsequent analysis, including the
multiplicativity of the independence and clique numbers. We then provide
alternative, elementary, and self-contained proofs of the
multiplicativity of the Lov\'{a}sz theta function and the fractional
Haemers number under lexicographic products, and of the equality between
the fractional and ordinary Lov\'{a}sz theta functions. Although these
three results are known, the proofs presented here are included for
their simplicity and because they provide a self-contained foundation
for the subsequent study of Shannon capacity.

The principal new results concern the Shannon capacity. We first prove that
it is supermultiplicative under the lexicographic product:
\begin{align}
\label{eq: supermultiplicative Theta}
\Theta(G\circ H) \geq \Theta(G) \; \Theta(H).
\end{align}
Because the lexicographic product is generally noncommutative, the
graphs \(G\circ H\) and \(H\circ G\) need not be isomorphic and must be
considered separately. After the natural identification of their vertex
sets, \(G\boxtimes H\) is a spanning subgraph of both \(G\circ H\) and
\(H\circ G\). We consequently obtain 
\begin{align}
\max \bigl\{ \Theta(G \circ H), \; \Theta(H\circ G) \bigr\} \leq \Theta(G\boxtimes H).
\end{align}

A central result of this paper shows that the supermultiplicativity
inequality for \(\Theta\) in
\eqref{eq: supermultiplicative Theta} can be strict. We explicitly
construct three countably infinite families of lexicographic powers
based on the Schl\"{a}fli graph, the McLaughlin graph, and its second
subconstituent. Pairing each member of these families with its complement
yields strict supermultiplicativity of Shannon capacity under the
lexicographic product, and arbitrarily large multiplicative gaps are
obtained within each family; see
Propositions~\ref{proposition: independence numbers of squared lexicographic products}
and~\ref{proposition: explicit graphs satisfying strict graph complement condition},
as well as
Theorem~\ref{theorem: strict supermultiplicativity of capacity}.
These results show, in particular, that the behavior of the Shannon capacity
for lexicographic products differs fundamentally from that of the
Lov\'{a}sz theta function and the fractional Haemers number, both of
which are multiplicative graph invariants.

The comparison with the strong product, together with the
multiplicativity of the Lov\'{a}sz theta function and the fractional
Haemers number under strong products, yields upper and lower bounds on
the capacities of both ordered lexicographic products, as well as
sufficient conditions for determining these capacities exactly. In
particular, if Shannon capacity is multiplicative under the strong
product of \(G\) and \(H\), then it is also multiplicative under both
\(G\circ H\) and \(H\circ G\). Exact formulas also follow when the
Shannon capacities of the individual factors equal their respective
Lov\'{a}sz theta numbers.

The Lov\'{a}sz theta function \cite{Lovasz79} and the fractional
Haemers number \cite{BukhC19} provide two potentially different upper
bounds on the Shannon capacity of a lexicographic product. We show that
these bounds are incomparable: neither dominates the other in general,
and either may be sharper, depending on the graphs under consideration.
We apply the resulting bounds and exactness criteria to lexicographic
products involving Kneser graphs, their complements, and \(q\)-analogues
of Kneser graphs.

We also prove that a complete outer factor preserves the Shannon
capacity of the inner factor; equivalently, the join of any finite
number of identical copies of a graph has the same Shannon capacity as
the graph itself. Finally, we determine the capacities of several
families of iterated lexicographic powers, including powers of
self-complementary graphs that are vertex-transitive or strongly
regular.

Recent developments further illustrate the range of parameters,
methods, and constructions in zero-error information theory. Examples
include the \(\rho\)-capacity motivated by zero-error broadcasting
\cite{HuS17}, spectral and rank-type bounds for distance powers of
graphs \cite{AbiadDF2026}, Shannon capacity and the Lov\'{a}sz theta
number under the Mycielski construction \cite{CsonkaS24},
quantum-mechanical and finite-automata methods for bounding graph
capacity \cite{Meiburg2025}, and zero-error capacities of channels with
memory \cite{CaoChenBai2025}.

Lexicographic products occur in a variety of graph-theoretic
constructions and provide a natural framework for studying the behavior
of graph invariants under substitution and composition. More broadly,
graph sums and products have been studied, among others, in
\cite{LaviSason2026,Schrijver23}, while the categorical product and
related multiplicativity phenomena have been studied in
\cite{Simonyi21}. Asymptotic graph parameters provide a systematic
language for analyzing graph operations and their associated
asymptotic rates \cite{Roberson2016,Vrana2021}.

The remainder of the paper is organized as follows.
Section~\ref{section: Preliminaries} reviews the necessary background on
lexicographic and strong graph products, the Lov\'{a}sz theta function,
the fractional Haemers number, the Shannon capacity, and relevant
properties of Kneser, generalized \(q\)-Kneser, and Paley graphs.
Section~\ref{section: Lovasz theta-function under a lexicographic product of graphs}
establishes the multiplicativity of the Lov\'{a}sz theta function under
lexicographic products and gives an alternative, elementary, and
self-contained proof that the fractional Lov\'{a}sz theta function
coincides with the ordinary one.
Section~\ref{section: Fractional Haemers number under a lexicographic product of graphs}
gives an alternative, elementary, and self-contained proof of the
corresponding multiplicativity result for the fractional Haemers number.
Section~\ref{section: Shannon capacity under a lexicographic product of graphs}
contains the new results on Shannon capacity, including the general
bounds, strict-supermultiplicativity constructions, exactness criteria,
and applications to several families of lexicographic graph products
and iterated lexicographic powers.
Finally, Section~\ref{section: an open problem} formulates an open
problem concerning the Shannon capacity under lexicographic and strong
products.

\section{Preliminaries}
\label{section: Preliminaries}

\subsection{Lexicographic and strong products of graphs}
\label{subsection: Lexicographic and strong products of graphs}

This subsection recalls the lexicographic and strong products of graphs
and compares some of their basic structural properties. We first review
their definitions and the relation between their edge sets, and then
examine the behavior of the independence and clique numbers under these
two products. These properties will be used throughout the paper.

\begin{definition}[Lexicographic product]
\label{definition: Lexicographic product}
{\em Let \(G\) and \(H\) be finite simple graphs. The \emph{lexicographic product}
of \(G\) and \(H\), denoted by \(G[H]\) or \(G\circ H\), is the graph with
vertex set
\begin{align}
V(G \circ H)=V(G)\times V(H),
\end{align}
in which two distinct vertices \((g,h)\) and \((g',h')\) are adjacent if and
only if
\begin{align}
\bigl(g \sim_G g' \bigr) \; \vee \; \bigl(g=g' \; \wedge \; h\sim_H h'\bigr).
\end{align}
For a graph \(G\) and an integer \(n \geq 1\), the \emph{\(n\)-fold lexicographic 
power} of \(G\), denoted by \(G^{\circ n}\), is the lexicographic product
of \(n\) copies of \(G\).}
\end{definition}

For completeness, we recall several elementary structural properties of
the lexicographic product. It is associative (see \cite[Proposition~5.11]{HammackIK11}):
\begin{align}
(F \circ G) \circ H \cong F \circ(G \circ H),
\end{align}
but it is not commutative in general; that is, \(G\circ H\) and
\(H\circ G\) need not be isomorphic (see \cite[Section~10.3]{HammackIK11}). 
It is right-distributive over disjoint unions:
\begin{align}
(F+G) \circ H \cong (F \circ H)+(G \circ H),
\end{align}
where \(+\) denotes disjoint union, but there is no corresponding 
left-distributive law. For example (see \cite[p.~57]{HammackIK11}), 
\[
K_2 \circ (K_1 + K_1) = C_4, \qquad (K_2 \circ K_1) + (K_2 \circ K_1) = K_2 + K_2,
\]
where \(C_n\) and \(K_n\) denote, respectively, the cycle graph on \(n \geq 3\) vertices 
and the complete graph on \(n \geq 1\) vertices. 

\smallskip 
\noindent 
It is easily verified that complementation commutes with the lexicographic product:
\begin{align}
\label{eq: complement of lexicographic product}
\overline{G \circ H} = \overline{G} \circ \overline{H}.
\end{align}
Consequently, the lexicographic product of two self-complementary graphs
is itself self-complementary. Lexicographic products and their structural 
properties are further studied in \cite[Chapter~10]{HammackIK11}.

\smallskip
The Shannon capacity of a graph, which is the main subject of this paper,
is defined in terms of strong graph powers.
\begin{definition}[Strong product]
\label{definition: Strong product}
{\em Let \(G\) and \(H\) be finite simple graphs. The \emph{strong product}
of \(G\) and \(H\), denoted by \(G\boxtimes H\), is the graph with
vertex set
\begin{align}
V(G\boxtimes H)=V(G) \times V(H),
\end{align}
in which two distinct vertices \((g,h)\) and \((g',h')\) are adjacent if
and only if
\begin{align}
\hspace*{-0.2cm} \bigl(g=g'\ \wedge \ h\sim_H h'\bigr)
\; \vee \; 
\bigl(g\sim_G g'\ \wedge \ h=h'\bigr)
\; \vee \; 
\bigl(g\sim_G g'\ \wedge \ h\sim_H h'\bigr).
\end{align}
Equivalently, if \((g,h)\neq(g',h')\), then 
\begin{align}
\hspace*{-0.2cm} (g,h) \sim_{G \boxtimes H}(g',h') \, \Longleftrightarrow \,
\bigl(g=g'\ \vee \ g\sim_G g'\bigr) \ \wedge \
\bigl(h=h'\ \vee \ h\sim_H h'\bigr).
\end{align}
For a graph \(G\) and an integer \(n \geq 1\), the \emph{\(n\)-fold strong
power} of \(G\), denoted by \(G^{\boxtimes n}\), is the strong product
of \(n\) copies of \(G\).}
\end{definition}

The two graph products in Definitions~\ref{definition: Lexicographic product} 
and~\ref{definition: Strong product} have the same vertex set, but their 
adjacency relations differ: if \(g\sim_G g'\), then the lexicographic product
places no restriction on \(h\) and \(h'\), whereas the strong product
requires that \(h=h'\) or \(h\sim_H h'\). This observation is formalized
in the following lemma, which is used later in this paper. 

\begin{lemma}
\label{lemma: spanning subgraph}
Let \(G\) and \(H\) be finite simple graphs. Then the strong product
\(G \boxtimes H\) is a spanning subgraph of the lexicographic product
\(G \circ H\).
\end{lemma}

\begin{proof}
Both \(G \boxtimes H\) and \(G \circ H\) have the same vertex set
\(V(G)\times V(H)\). It therefore suffices to prove that
\begin{align}
\label{eq1: 18.08.26}
E(G\boxtimes H)\subseteq E(G\circ H).
\end{align}
Recall that two distinct vertices \((g,h)\) and \((g',h')\) are
adjacent in \(G\circ H\) if and only if
\[
\{g,g'\}\in E(G)
\quad \vee \quad
\bigl(g=g' \; \wedge \; \{h,h'\}\in E(H)\bigr).
\]
Let \((g,h)\) and \((g',h')\) be adjacent in \(G\boxtimes H\).
By the definition of the strong product, one of the following holds:
\[
\begin{aligned}
&g=g' && \wedge &&\{h,h'\}\in E(H),\\
&\{g,g'\}\in E(G) && \wedge &&h=h',\\
&\{g,g'\}\in E(G) && \wedge &&\{h,h'\}\in E(H).
\end{aligned}
\]
In the first case, \((g,h)\) and \((g',h')\) are adjacent in
\(G\circ H\) because \(g=g'\) and \(\{h,h'\} \in E(H)\). In the other two
cases, they are adjacent in \(G\circ H\) because \(\{g,g'\} \in E(G)\).
Consequently, \eqref{eq1: 18.08.26} holds. 
Since the two products have the same vertex set, \(G\boxtimes H\) is
a spanning subgraph of \(G\circ H\).
\end{proof}

\begin{proposition}
\label{proposition: independence and clique numbers of lexicographic product}
Let \(G\) and \(H\) be finite simple graphs. Then the independence and
clique numbers of their lexicographic product satisfy
\begin{align}
\alpha(G\circ H)
&=\alpha(G)\,\alpha(H),
\label{eq: independence number of lexicographic product}\\
\omega(G\circ H)
&=\omega(G)\,\omega(H).
\label{eq: clique number of lexicographic product}
\end{align}
\end{proposition}

For completeness, we prove 
Proposition~\ref{proposition: independence and clique numbers of lexicographic product}.
The independence-number identity
\eqref{eq: independence number of lexicographic product} was proved in
\cite{GellerS75} and is also stated in
\cite[Problem~27.1]{HammackIK11}, where the corresponding clique-number
identity \eqref{eq: clique number of lexicographic product} is not mentioned. 
We include a reformulation of the proof of
\eqref{eq: independence number of lexicographic product}, together with
a proof of \eqref{eq: clique number of lexicographic product}.
  
\begin{proof}
We first prove~\eqref{eq: independence number of lexicographic product}.
Let \(A\) and \(B\) be maximum independent sets in \(G\) and \(H\),
respectively. We claim that \(A\times B\) is an independent set in
\(G\circ H\). Indeed, consider two distinct vertices
\((g,h),(g',h')\in A\times B\). If \(g\neq g'\), then
\(g\not\sim_G g'\) because \(A\) is independent. If \(g=g'\), then
\(h\not\sim_H h'\) because \(B\) is independent. Thus, the two vertices
are not adjacent in \(G\circ H\), and consequently
\[
\alpha(G\circ H)
\geq |A\times B|
=\alpha(G)\,\alpha(H).
\]
For the reverse inequality, let \(I\) be an independent set in
\(G\circ H\), and define
\[
I_g \coloneqq \{h\in V(H):(g,h)\in I\}, \qquad g\in V(G).
\]
Each \(I_g\) is an independent set in \(H\), and therefore
\(
|I_g|\leq\alpha(H).
\)
Moreover, the set
\[
P \coloneqq \{g\in V(G):I_g\neq\varnothing\}
\]
is independent in \(G\). Indeed, if distinct \(g,g'\in P\) were adjacent
in \(G\), then every vertex of \(\{g\}\times I_g\) would be adjacent in
\(G\circ H\) to every vertex of \(\{g'\}\times I_{g'}\), contradicting
the independence of \(I\). Hence, \(|P|\leq\alpha(G)\), and
\[
|I|
=\sum_{g\in P}|I_g|
\leq |P|\,\alpha(H)
\leq\alpha(G)\,\alpha(H).
\]
Taking the maximum over all independent sets \(I\) proves the reverse inequality 
\[
\alpha(G\circ H) \leq \alpha(G)\,\alpha(H),
\]
which then establishes equality \eqref{eq: independence number of lexicographic product}.

We next prove~\eqref{eq: clique number of lexicographic product}.  
Combining identities \eqref{eq: complement of lexicographic product}
and \eqref{eq: independence number of lexicographic product} yields 
\begin{align}
\omega(G \circ H) &= \alpha( \overline{G \circ H}) \nonumber \\
&= \alpha( \overline{G} \circ \overline{H}) \nonumber \\
&= \alpha( \overline{G}) \, \alpha( \overline{H}) \nonumber \\
&= \omega(G)\,\omega(H),
\end{align}
where the first and last equalities follow from the relation
\( \alpha(F) = \omega(\overline{F}) \), valid for every graph \(F\). 
\end{proof}

\begin{remark}
{\em
In contrast to
Proposition~\ref{proposition: independence and clique numbers of lexicographic product},
the independence number is, in general, only supermultiplicative under the
strong product:
\[
\alpha(G\boxtimes H)\geq\alpha(G)\,\alpha(H).
\]
Indeed, the Cartesian product of an independent set in \(G\) and an
independent set in \(H\) is independent in \(G\boxtimes H\), whereas the
inequality can be strict. For example, label the vertices of \(C_5\) by
the elements of \(\mathbb{Z}_5\), with two vertices adjacent when their
difference is \(\pm1\) modulo \(5\). Then
\[
S \coloneqq \bigl\{(i,2i):i\in\mathbb{Z}_5\bigr\}
=\bigl\{(0,0),(1,2),(2,4),(3,1),(4,3)\bigr\}
\]
is an independent set in \(C_5\boxtimes C_5\). In fact,
\[
5=\alpha(C_5\boxtimes C_5)>\alpha(C_5)^2=4.
\]

The clique number, on the other hand, is multiplicative under the strong
product:
\[
\omega(G\boxtimes H)=\omega(G)\,\omega(H).
\]
To see this, let \(A\) and \(B\) be maximum cliques in \(G\) and \(H\),
respectively. For any two distinct vertices
\((g,h),(g',h')\in A\times B\), we have
\[
g=g' \ \text{or}\ g\sim_G g',
\qquad
h=h' \ \text{or}\ h\sim_H h'.
\]
Thus, \(A\times B\) is a clique in \(G\boxtimes H\), and hence
\[
\omega(G\boxtimes H)
\geq |A\times B|
=\omega(G)\,\omega(H).
\]
For the reverse inequality, let \(C\) be a clique in \(G\boxtimes H\), and let
\begin{align*}
& P_G \coloneqq \bigl\{g\in V(G):(g,h)\in C
\text{ for some }h\in V(H)\bigr\}, \\
& P_H \coloneqq \bigl\{h\in V(H):(g,h)\in C
\text{ for some }g\in V(G)\bigr\}
\end{align*}
be its coordinate projections. The set \(P_G\) is a clique in \(G\).
Indeed, if \(g,g'\in P_G\) are distinct, then there exist \(h,h'\in V(H)\)
such that \((g,h),(g',h')\in C\). Since these vertices are adjacent in
\(G\boxtimes H\) and \(g\neq g'\), it follows that \(g\sim_G g'\).
Similarly, \(P_H\) is a clique in \(H\). Therefore,
\[
|P_G|\leq\omega(G),
\qquad
|P_H|\leq\omega(H).
\]
Since \(C\subseteq P_G\times P_H\), we obtain
\[
|C|
\leq |P_G|\,|P_H|
\leq\omega(G)\,\omega(H).
\]
Maximizing over all cliques \(C\) gives the reverse inequality and proves
the claimed equality.

Thus, the distinction between the two products lies in the behavior of the
independence number: it is multiplicative under the lexicographic product,
but in general only supermultiplicative under the strong product.
}
\end{remark}

\begin{remark}
\label{remark: commutativity}
{\em Although it follows from 
Proposition~\ref{proposition: independence and clique numbers of lexicographic product} that 
\[
\alpha(G\circ H)=\alpha(H\circ G)
\quad\text{and}\quad
\omega(G\circ H)=\omega(H\circ G),
\]
the graphs \(G\circ H\) and \(H\circ G\) need not be isomorphic.
In contrast, the strong product is commutative up to graph isomorphism:
\[
    G\boxtimes H \cong H \boxtimes G.
\]
Indeed, the coordinate-switching map
\[
    V(G\boxtimes H)\longrightarrow V(H\boxtimes G),
    \qquad (g,h)\longmapsto (h,g),
\]
is a graph isomorphism.}
\end{remark}

\subsection{Lov\'{a}sz theta function}
\label{subsection: Lovasz theta-function}

The Lov\'{a}sz \(\vartheta\)-function is a graph invariant introduced in
\cite{Lovasz79} and defined in terms of orthogonal representations of the graph.
It can be computed efficiently and provides computable bounds on 
several graph invariants whose exact computation is NP-hard. It has been studied, 
e.g., in \cite{CsonkaS24, LaviSason2026, Lovasz79, Lovasz19, Sason23, Sason24}. The 
Lov\'{a}sz $\vartheta$-function of a finite simple graph $G$ can be expressed as a solution 
of a semidefinite programming (SDP) problem. To that end, let ${\bf{A}} = (A_{i,j})$ 
be the $n \times n$ adjacency matrix of $G$ with $n \triangleq |V(G)|$. The 
Lov\'{a}sz $\vartheta$-function $\vartheta(G)$ can be expressed by the following 
convex optimization problem:
\vspace*{0.1cm}
\begin{eqnarray}
\label{eq: SDP problem - Lovasz theta-function}
\mbox{\fbox{$
\begin{array}{l}
\text{maximize} \; \; \mathrm{Tr}({\bf{B}} \, {\bf{J}}_n)  \\
\text{subject to} \\
\begin{cases}
{\bf{B}} \succeq 0, \\
\mathrm{Tr}({\bf{B}}) = 1, \\
A_{i,j} = 1  \; \Rightarrow \;  B_{i,j} = 0, \quad i,j \in [n].
\end{cases}
\end{array}$}}
\end{eqnarray}
The SDP formulation in \eqref{eq: SDP problem - Lovasz theta-function} yields
the existence of an algorithm that computes $\vartheta(G)$, for every graph $G$, with a
precision of $r$ decimal digits, and a computational complexity that is polynomial in $n$ and $r$.
Thus, the Lov\'{a}sz $\vartheta$-function can be computed in polynomial time in $n$, and by \cite{Lovasz79},
it is an upper bound on the Shannon capacity. By contrast, determining the Shannon capacity involves an
infinite sequence of independence numbers, each of which is NP-hard to compute.

\subsection{Fractional Haemers number}
\label{subsection: Fractional Haemers number}

The fractional Haemers number is a graph invariant studied in 
\cite{Blasiak13, BukhC19, Csonka26, Fritz21, Roberson2016}.
We provide here some essential background. 

\begin{definition}[Fractional Haemers number]
\label{definition: Fractional Haemers number}
{\em Let \(G=(V,E)\) be a finite simple graph, let \(\mathbb{F}\) be a field,
and write \([d] \coloneqq \{1,\ldots,d\}\) for an integer \(d \geq 1\). A
matrix
\[
M\in\mathbb{F}^{(V\times[d])\times(V\times[d])}
\]
is called a \emph{\(d\)-representation of \(G\) over \(\mathbb{F}\)} if
its \(d\times d\) blocks \(M_{u,v}\), indexed by \(u,v\in V\), satisfy
\[
M_{v,v}=I_d  \qquad \text{for every }v\in V,
\]
and
\[
M_{u,v}=M_{v,u}=O_d
\qquad\text{whenever }u\neq v
\text{ and }\{u,v\}\notin E,
\]
where \(I_d\) and \(O_d\) denote the \(d\times d\) identity and zero
matrices, respectively.
For an integer \(d \geq 1\), let \(\mathcal{M}_d(G; \mathbb{F}) \) denote 
the set of all \(d\)-representations of \(G\) over \( \mathbb{F} \). 
\begin{enumerate}
\item \hspace*{-0.4cm} The \emph{fractional Haemers number} of \(G\) over \(\mathbb{F}\) is
defined by
\begin{align}
\label{eq: fractional Haemers number}
\mathcal{H}_{f}(G;\mathbb{F}) \coloneqq
\inf_{d \geq 1} \, \min_{M \in \mathcal{M}_d(G; \mathbb{F})} 
\frac{\operatorname{rank}_{\mathbb{F}}(M)}{d}.
\end{align}
\item \hspace*{-0.4cm} Restricting the optimization to \(d=1\) gives the
\emph{Haemers number} \cite{Haemers79}, defined by 
\begin{align}
\label{eq: Haemers number}
\mathcal{H}(G;\mathbb{F}) \coloneqq
\min_{M \in \mathcal{M}_1(G; \mathbb{F})}  
\operatorname{rank}_{\mathbb{F}}(M).
\end{align}
\end{enumerate}}
\end{definition}

\begin{definition}[Fractional Lov\'{a}sz theta function]
\label{definition: Fractional Lovasz theta function}
{\em Let \(G\) be a finite simple graph. The \emph{fractional Lov\'{a}sz theta
function} of \(G\) is defined by
\begin{align}
\label{eq: Fractional Lovasz theta function}
\vartheta_f(G) \coloneqq \inf_{d \geq 1}
\frac{\vartheta\bigl(G\circ\overline{K_d}\bigr)}{d},
\end{align}
where \(\overline{K_d}\) is the edgeless graph on \(d\) vertices.}
\end{definition}

\begin{proposition}
\label{proposition: fractional Haemers number}
Let \(G\) be a finite simple graph, and let \(\mathbb{F}\) be a field.
Then
\begin{align}
\label{eq1: 06.08.26}
\mathcal{H}_f(G;\mathbb{F})
= \inf_{d \geq 1} \frac{\mathcal{H}\bigl(G \circ \overline{K_d}; \mathbb{F}\bigr)}{d},
\end{align}
where \(\overline{K_d}\) denotes the edgeless graph on \(d\) vertices.
\end{proposition}

\begin{proof}
For every integer \(d \geq 1\), the vertex set of
\(G \circ \overline{K_d}\) is \(V(G) \times [d]\). Two distinct vertices
\((u,i)\) and \((v,j)\) are nonadjacent in
\(G \circ \overline{K_d}\) precisely when either \(u=v\), or
\(u \neq v\) and \(\{u,v\} \notin E(G)\).
Consequently, a matrix fitting \(G \circ \overline{K_d}\) over
\(\mathbb{F}\), viewed as a block matrix whose \(d \times d\) blocks are
indexed by \(u,v \in V(G)\), has diagonal blocks equal to \(I_d\) and
zero blocks whenever \(u \neq v\) and \(\{u,v\} \notin E(G)\). Thus, the
matrices fitting \(G \circ \overline{K_d}\) are precisely the
\(d\)-representations of \(G\) over \(\mathbb{F}\). It follows that
\begin{align}
\label{eq2: 06.08.26}
\mathcal{H}\bigl(G \circ \overline{K_d};\mathbb{F}\bigr)
= \min_{M \in \mathcal{M}_d(G; \mathbb{F})} \operatorname{rank}_{\mathbb{F}}(M).
\end{align}
Dividing both sides of \eqref{eq2: 06.08.26} by \(d\) and taking the infimum 
over all integers \(d \geq 1\), we obtain \eqref{eq1: 06.08.26} by \eqref{eq: fractional Haemers number}.
\end{proof}

\subsection{Shannon capacity of graphs}
\label{subsection: Shannon capacity of graphs}

\begin{definition}[Shannon capacity of a graph]
\label{definition: Shannon capacity of a graph}
{\em The \emph{Shannon capacity} of a finite simple graph \(G\) is defined as  
\begin{align}
\label{eq: Shannon capacity}
\Theta(G) \coloneqq \sup_{n \geq 1}
\sqrt[n]{\alpha\!\left(G^{\boxtimes n}\right)}
= \lim_{n \to \infty} \sqrt[n]{\alpha\!\left(G^{\boxtimes n}\right)},
\end{align}
where the existence of the limit and the last equality follow from Fekete's lemma and the   
supermultiplicativity of the independence number under strong products.}
\end{definition}

The following result by Lov\'{a}sz provides an easy-to-compute upper bound on the Shannon capacity of graphs
(see \cite[Theorem~1]{Lovasz79}).
\begin{theorem}
Let \(G\) be a finite simple graph. Then its Shannon capacity is upper bounded by the Lov\'{a}sz theta function:
\begin{align}
\Theta(G) \leq \vartheta(G).
\end{align}
\end{theorem}

The next upper bound on the Shannon capacity was first proved by
Blasiak \cite{Blasiak13}; see also \cite[Theorem~1]{BukhC19}.
\begin{theorem}
\label{theorem: Shannon capacity is upper bounded by the fractional Haemers number}
Let \(G\) be a finite simple graph, and let \(\mathbb{F}\) be a field.
Then the Shannon capacity of \(G\) is upper bounded by its fractional
Haemers number over \(\mathbb{F}\); that is,
\begin{align}
\Theta(G)\leq \mathcal{H}_f(G;\mathbb{F}).
\end{align}
Moreover,
\begin{align}
\alpha(G) \leq \Theta(G) \leq \mathcal{H}_f(G;\mathbb{F})
\leq \mathcal{H}(G;\mathbb{F}).
\end{align}
\end{theorem}

\begin{proof}
For every finite graph \(K\),
\[
\alpha(K)\leq \mathcal{H}_f(K;\mathbb{F}).
\]
Applying this inequality to the \(n\)-th strong power of \(G\), and
using the multiplicativity of the fractional Haemers number under the
strong product (see \cite[Theorem~3]{BukhC19}), gives 
\[
\alpha\bigl(G^{\boxtimes n}\bigr)
\leq
\mathcal{H}_f\bigl(G^{\boxtimes n};\mathbb{F}\bigr)
=
\mathcal{H}_f(G;\mathbb{F})^n.
\]
Taking \(n\)-th roots and then the supremum over \(n \geq 1\), we obtain
\[
\Theta(G)
= \sup_{n \geq 1} \sqrt[n]{\alpha\bigl(G^{\boxtimes n}\bigr)}
\leq \mathcal{H}_f(G;\mathbb{F}).
\]
Finally, restricting the definition of
\(\mathcal{H}_f(G;\mathbb{F})\) to \(d=1\) yields
\[
\mathcal{H}_f(G;\mathbb{F}) \leq \mathcal{H}(G;\mathbb{F}),
\]
while \(\alpha(G)\leq \Theta(G)\) follows by taking \(n=1\) in the 
middle term of \eqref{eq: Shannon capacity}.
\end{proof}

In general, the fractional Haemers number and the Lov\'{a}sz theta
function are incomparable. That is, depending on the graph and the
underlying field, either one may provide the stronger upper bound on the
Shannon capacity. In particular, Bukh and Cox established the following
strict separation.

\begin{theorem}[Bukh and Cox, {\cite[Theorem~2]{BukhC19}}]
Let \(\mathbb{F}\) be a field of nonzero characteristic. Then there
exists a finite graph \(G=G(\mathbb{F})\) such that, for every
field \(\mathbb{F}'\),
\begin{align}
\mathcal{H}_f(G;\mathbb{F})
< \min\bigl\{ \mathcal{H}(G;\mathbb{F}'),\vartheta(G) \bigr\}.
\end{align}
In particular,
\begin{align}
\mathcal{H}_f(G;\mathbb{F})<\vartheta(G).
\end{align}
\end{theorem}

A further explicit separation was obtained by Hu, Tamo, and Shayevitz
(see \cite[Example~6 and Remark~3]{HuTS18}). Let \(S\) denote the
complement of the Schl\"afli graph, and let
\begin{align}
G \coloneqq S+7C_5,
\end{align}
where \(+\) denotes the disjoint union. They showed that
\[
\mathcal{H}_f(G;\mathbb{F}_{11})
\leq \tfrac{49}{2} < 9+7\sqrt{5} = \vartheta(G)
< 28 = \mathcal{H}(G),
\]
where
\begin{align}
\mathcal{H}(G) \coloneqq \min_{\mathbb{F}'} \mathcal{H}(G;\mathbb{F}')
\end{align}
is the field-independent Haemers number. Thus, this graph satisfies
\begin{align}
\mathcal{H}_f(G;\mathbb{F}_{11}) < \vartheta(G) < \mathcal{H}(G).
\end{align}

\subsection{Kneser graphs}
\label{subsection: Kneser graphs}

Kneser graphs form an important family of highly symmetric graphs in
algebraic and extremal graph theory. 
\begin{definition}
\label{definition: Kneser graphs}
{\em The Kneser graph \(\mathrm{KG}(n,k)\), defined for integers \(n \geq 2k\), has as its
vertices all \(k\)-element subsets of \([n] \coloneqq \{1,\ldots,n\}\), with
two vertices adjacent if and only if the corresponding subsets are disjoint.}
\end{definition}
By Definition~\ref{definition: Kneser graphs}, \(\mathrm{KG}(n,k)\) has \(\binom{n}{k}\) 
vertices and is \(\binom{n-k}{k}\)-regular. Independent sets in \(\mathrm{KG}(n,k)\) 
correspond precisely to intersecting families of \(k\)-element subsets of \([n]\).
Consequently, the Erd\H{o}s--Ko--Rado theorem gives
\begin{align}
\label{eq: alpha of Kneser graphs}
\alpha\bigl(\mathrm{KG}(n,k)\bigr) = \binom{n-1}{k-1}, \qquad n \geq 2k.
\end{align}
Moreover, the Shannon capacity of \(\mathrm{KG}(n,k)\) equals its independence number 
(see \cite[Theorem~13]{Lovasz79}), and hence
\begin{align}
\label{eq: capacity of Kneser graphs}
\Theta\bigl(\mathrm{KG}(n,k)\bigr) = \binom{n-1}{k-1}, \qquad n \geq 2k.
\end{align}

\begin{example}
\label{example: Kneser graph Petersen}
{\em The Kneser graph \(KG(5,2)\) is isomorphic to the well-known Petersen graph.
Its vertices are the ten 2-element subsets of \([5]\), with two
vertices adjacent if and only if the corresponding subsets are
disjoint. For example, the vertex \(\{1,2\}\) is adjacent precisely to
\(\{3,4\}\), \(\{3,5\}\), and \(\{4,5\}\). Thus, \(KG(5,2)\) is a
3-regular graph on 10 vertices. By \eqref{eq: alpha of Kneser graphs}, 
its independence number is \(4\); for example,
\(
\bigl\{\{1,2\},\{1,3\},\{1,4\},\{1,5\}\bigr\}
\)
is a maximum independent set. Moreover, by \eqref{eq: capacity of Kneser graphs}, 
its Shannon capacity is also equal to \(4\), that is,
\[
\Theta\bigl(KG(5,2)\bigr)=\alpha\bigl(KG(5,2)\bigr)=4.
\]}
\end{example}

\noindent 
We next define $q$-analogues of Kneser graphs and Gaussian binomial coefficients; the interested 
reader is referred to \cite[Chapter~9]{GodsilMeagher2016} for more information.
\begin{definition}\label{definition:q-gaussian_coefficient}
{\em Let \(n,k,m \geq 1\) be integers with \(k\leq n\), let \(p\) be a prime, and
let \(q \coloneq p^m\). Let \(\mathbb{F}_q\) denote the finite field of order \(q\). 
The \emph{Gaussian binomial coefficient}, denoted by $\gbinom{n}{k}{q}$, is defined as 
\begin{align}\label{eq:q-gaussian_coefficient}
\gbinom{n}{k}{q} \coloneqq \frac{(q^n - 1)(q^n-q)\cdots(q^n - q^{k-1})}{(q^k - 1)(q^k - q)\cdots (q^k - q^{k-1})}.
\end{align}}
\end{definition}
The Gaussian coefficient admits the following combinatorial interpretation.
Let $V$ be an $n$-dimensional vector space over $\mathbb{F}_q$. Then, the number of distinct $k$-dimensional subspaces of $V$
is equal to $\gbinom{n}{k}{q}$. 
Allowing \(q\) to vary continuously and taking \(q\to 1\) gives
\begin{align}
\lim_{q \to 1} \, \gbinom{n}{k}{q} = \frac{n}{k} \cdot \frac{n-1}{k-1} \cdots \frac{n-(k-1)}{k-(k-1)} 
= \binom{n}{k},
\end{align}
thus converging to the binomial coefficient.

\begin{definition}[$q$-Kneser graphs]
\label{definition:q_kneser_graph}
{\em Let $V(n,q)$ denote the $n$-dimensional vector space over the finite field $\mathbb{F}_q$, where $q$ is a prime power.
The $q$-Kneser graph $\mathrm{KG}_q(n,k)$ has as its vertices the $k$-dimensional subspaces of $V(n,q)$, and any two vertices
are adjacent if and only if their intersection is the zero subspace.}
\end{definition}

\noindent
By \cite[Theorem~6.6]{LaviSason2026},
\begin{align}
\Theta\bigl(\mathrm{KG}_q(n,k)\bigr)
=\alpha\bigl(\mathrm{KG}_q(n,k)\bigr)
=\gbinom{n-1}{k-1}{q}, \qquad n \geq 2k.
\end{align}

\subsection{Paley graphs}
\label{subsection: Paley graphs}

\begin{definition}
\label{definition: Paley graphs}
{\em Let \(n \geq 5\) be a prime power such that
\(
n\equiv 1 \pmod 4.
\)
The \emph{Paley graph} \(P(n)\) is the graph with vertex set
\(
V(P(n))=\mathbb{F}_n,
\)
in which two distinct vertices \(x,y\in\mathbb{F}_n\) are adjacent if
and only if \(x-y\) is a nonzero square in \(\mathbb{F}_n\).}
\end{definition}

The congruence condition \(n\equiv 1\pmod 4\) ensures that \(-1\) is a
square in \(\mathbb{F}_n\). Hence, \(x-y\) is a square if and only if
\(y-x\) is a square, so the adjacency relation is symmetric.
Paley graphs are self-complementary, vertex-transitive, and strongly
regular. More precisely, \(P(n)\) has parameters
\[
\operatorname{srg}\left(
n,\frac{n-1}{2},\frac{n-5}{4},\frac{n-1}{4}
\right).
\]
Thus, \(P(n)\) is a regular graph on \(n\) vertices of degree
\(\frac{n-1}{2}\); every pair of adjacent vertices has exactly
\(\lambda=\frac{n-5}{4}\) common neighbors, whereas every pair 
of nonadjacent vertices has exactly \( \mu=\frac{n-1}{4} \)
common neighbors. 

\begin{example}
\label{example: Paley graph P5}
\emph{The smallest Paley graph, \(P(5)\), is isomorphic to the  
cycle \(C_5\). Indeed, the nonzero squares in
\(\mathbb{F}_5=\{0,1,2,3,4\}\) are \(1\) and \(4=-1\). Hence, two
distinct vertices \(x,y\in\mathbb{F}_5\) are adjacent if and only if
\(
x-y\equiv \pm 1 \pmod 5,
\)
which gives the 5-length cycle
\(
0\mathbin{-}1\mathbin{-}2\mathbin{-}3\mathbin{-}4\mathbin{-}0.
\)
Accordingly, \(P(5)\) is self-complementary and vertex-transitive, and
it is strongly regular with parameters
\(
\operatorname{srg}(5,2,0,1).
\)
In particular, every pair of adjacent vertices has no common neighbor,
whereas every pair of nonadjacent vertices has exactly one common
neighbor. Moreover, \(P(5)\) is the unique triangle-free Paley graph:
for every admissible \(n>5\), each pair of adjacent vertices in \(P(n)\)
has
\(
\lambda=\frac{n-5}{4} \geq 1
\)
common neighbors and therefore lies in a triangle.}
\end{example}

\section{Lov\'{a}sz theta function for lexicographic graph products}
\label{section: Lovasz theta-function under a lexicographic product of graphs}

The next result was established in greater generality in \cite[Theorem~8.2]{Roberson2016}. 
We provide an alternative simplified proof for the specialized setting considered here.
\begin{theorem}
\label{theorem: Lovasz theta function of lexicographic product of graphs}
For all finite simple graphs \(G\) and \(H\),
\begin{align}
\label{eq: Lovasz theta function of lexicographic product of graphs}
\vartheta(G \circ H)=\vartheta(G) \, \vartheta(H).
\end{align}
\end{theorem}

\begin{proof}
We prove separately the two inequalities that give equality \eqref{eq: Lovasz theta function of lexicographic product of graphs}.

\medskip
\noindent
\textbf{Lower bound.}
We use the semidefinite-programming formulation \eqref{eq: SDP problem - Lovasz theta-function} of the Lov\'{a}sz $\vartheta$-function
\[
\vartheta(G) = \max\left\{ \sdpinner{J}{X}: X\succeq 0,\  \operatorname{Tr}(X)=1,\  X_{g,g'}=0\ \text{whenever } g \sim_G g' \right\},
\]
where \(J\) is the all-ones matrix, and \( \sdpinner{J}{X} = \operatorname{Tr}(JX) \).

Let \(X\) and \(Y\) be optimal solutions to the semidefinite programs defining \(\vartheta(G)\) and \(\vartheta(H)\), respectively, 
and let 
\[
Z \coloneqq X\otimes Y.
\]
Then
\[
Z \succeq 0, \qquad \operatorname{Tr}(Z) =\operatorname{Tr}(X) \, \operatorname{Tr}(Y)=1.
\]
If \((g,h)\sim_{G \circ H}(g',h')\), then either \(g\sim_G g'\), in which case
\(X_{g,g'}=0\), or \(g=g'\) and \(h \sim_H h'\), in which case \(Y_{h,h'}=0\).
Consequently, 
\[
Z_{(g,h),(g',h')}=X_{g,g'} \, Y_{h,h'}=0.
\]
Therefore, \(Z\) is feasible for the semidefinite program \eqref{eq: SDP problem - Lovasz theta-function} 
defining \(\vartheta(G \circ H)\): 
\begin{equation}
\label{eq1: 05.08.26}
\begin{aligned}
\vartheta(G\circ H)
=\max\bigl\{\, &
\sdpinner{J}{Z}: Z\succeq 0,\ 
\operatorname{Tr}(Z)=1, \\
& Z_{(g,h),(g',h')}=0 \ \text{whenever }\,
(g,h)\sim_{G\circ H}(g',h') \,\bigr\}.
\end{aligned}
\end{equation}
Moreover, since
\[J_{|V(G)| \, |V(H)|}=J_{|V(G)|}\otimes J_{|V(H)|},\] the standard identities
for Kronecker products give
\begin{align}
\sdpinner{J}{Z}
&= \operatorname{Tr}\bigl( (J_{|V(G)|}\otimes J_{|V(H)|})(X\otimes Y) \bigr) \nonumber \\
&= \operatorname{Tr}\bigl( (J_{|V(G)|} X) \otimes (J_{|V(H)|} Y) \bigr) \nonumber \\
&= \operatorname{Tr}(J_{|V(G)|}X) \, \operatorname{Tr}(J_{|V(H)|}Y) \nonumber \\
&= \sdpinner{J}{X}\sdpinner{J}{Y} \nonumber \\
\label{eq2: 05.08.26}
&= \vartheta(G) \, \vartheta(H),
\end{align}
which, by \eqref{eq1: 05.08.26} and \eqref{eq2: 05.08.26}, implies that
\begin{align}
\label{eq4: 06.08.26}
\vartheta(G \circ H) \geq \vartheta(G) \, \vartheta(H).
\end{align}

\medskip
\noindent
\textbf{Upper bound.}
By Lemma~\ref{lemma: spanning subgraph}, \(G\boxtimes H\) is a spanning
subgraph of \(G\circ H\). Since the Lov\'{a}sz theta function is
monotonically nonincreasing under the addition of edges, it follows that 
\[
\vartheta(G \circ H) \leq \vartheta(G \boxtimes H).
\]
Combining the last inequality with the identity $\vartheta(G \boxtimes H) = \vartheta(G) \, \vartheta(H)$ 
(see \cite[Theorem~7]{Lovasz79}) gives  
\begin{align}
\label{eq5: 06.08.26}
\vartheta(G \circ H) \leq \vartheta(G) \, \vartheta(H).
\end{align}
Combining inequalities \eqref{eq4: 06.08.26} and \eqref{eq5: 06.08.26}
proves equality \eqref{eq: Lovasz theta function of lexicographic product of graphs}.
\end{proof}

\begin{corollary}
For all integers \( m \geq 1\), the \(m\)-fold join of \(G\) with itself satisfies
\begin{align}
\vartheta\left(K_m \circ G \right) = \vartheta(G).
\end{align}
\end{corollary}
\begin{proof}
This holds by Theorem~\ref{theorem: Lovasz theta function of lexicographic product of graphs}
and since \(\vartheta(K_m)=1\).
\end{proof}

\noindent 
Applying Theorem~\ref{theorem: Lovasz theta function of lexicographic product of graphs}
yields the following result from \cite[Theorem~20]{BukhC19}, asserting the equality of 
the fractional and ordinary Lov\'{a}sz theta functions, together with a simplified proof.
\begin{corollary}
\label{corollary: equality of the fractional and ordinary Lovasz theta functions}
Let \(G\) be a finite simple graph. Then
\begin{align}
\label{eq: equality of the fractional and ordinary Lovasz theta functions}
\vartheta_f(G)=\vartheta(G).
\end{align}
\end{corollary}

\begin{proof}
By Definition~\ref{definition: Fractional Lovasz theta function},
Theorem~\ref{theorem: Lovasz theta function of lexicographic product of graphs},
and the identity 
\(
\vartheta(\overline{K_d})=d,
\)
where \(\overline{K_d}\) is the edgeless graph on \(d\) vertices, we have 
\begin{align*}
\vartheta_f(G)
=\inf_{d \geq 1} \frac{\vartheta\bigl(G\circ\overline{K_d}\bigr)}{d}  
=\inf_{d \geq 1} \frac{\vartheta(G) \, \vartheta(\overline{K_d})}{d}  
=\vartheta(G).
\end{align*}
This proves \eqref{eq: equality of the fractional and ordinary Lovasz theta functions}.
\end{proof}

\section{Fractional Haemers number for lexicographic graph products}
\label{section: Fractional Haemers number under a lexicographic product of graphs}

The following result establishes the multiplicativity of the fractional
Haemers number under the lexicographic product. It was stated by Fritz
in \cite[Proposition~4.13]{Fritz21}, where the lexicographic-product
argument is credited to Cox and described as an adaptation of the
argument of Bukh and Cox \cite{BukhC19}. The proof presented there is 
different and relies on previously established results. We next give an 
alternative, self-contained proof.
 
\begin{theorem}
\label{theorem: Multiplicativity of H_f under the lexicographic product}
Let \(G\) and \(H\) be finite simple graphs, and let \(\mathbb{F}\) be
a field. Then
\begin{align}
\label{eq3: 05.08.26}
\mathcal{H}_f(G\circ H;\mathbb{F})
= \mathcal{H}_f(G;\mathbb{F}) \, \mathcal{H}_f(H;\mathbb{F}).
\end{align}
\end{theorem}

\begin{proof}
We prove equality \eqref{eq3: 05.08.26} by establishing the two corresponding
inequalities separately.

\smallskip \noindent {\textbf{Upper bound.}}
We first prove an upper bound on $\mathcal{H}_f(G\circ H;\mathbb{F})$. 
Let $d, e \geq 1$ be arbitrary integers, let \(M\) be a \(d\)-representation of
\(G\) over \(\mathbb{F}\), and let \(N\) be an \(e\)-representation of
\(H\) over \(\mathbb{F}\). Define a block matrix \(P\), indexed by
\(V(G)\times V(H)\), by
\[
P_{(u,x),(v,y)}\coloneqq M_{u,v}\otimes N_{x,y}, \qquad \text{for all  } 
u,v \in V(G) \text {  and  } x,y \in V(H). 
\]
Each block of \(P\) has size \(de \times de\). Moreover, for every 
\( u \in V(G) \) and \( x \in V(H) \)
\[
P_{(u,x),(u,x)}
= M_{u,u}\otimes N_{x,x}
= I_d \otimes I_e
= I_{de}.
\]
Suppose that \((u,x)\) and \((v,y)\) are distinct and nonadjacent in
\(G\circ H\). Then either
\[
u\neq v \quad \text{and} \quad \{u,v\}\notin E(G),
\]
or
\[
u=v \quad \text{and} \quad \{x,y\}\notin E(H).
\]
In the first case, \(M_{u,v}=O_d\), and in the second case,
\(N_{x,y}=O_e\). Hence
\[
P_{(u,x),(v,y)}=O_{de}.
\]
Thus, \(P\) is a \(de\)-representation of \(G\circ H\). After suitable
permutations of its scalar rows and columns, \(P\) is equal to the
Kronecker product \(M\otimes N\). Since such permutations preserve
rank,
\begin{align*}
\operatorname{rank}_{\mathbb{F}}(P)
&= \operatorname{rank}_{\mathbb{F}}(M\otimes N) \nonumber \\
&= \operatorname{rank}_{\mathbb{F}}(M)\, \operatorname{rank}_{\mathbb{F}}(N),
\end{align*}
and consequently, since \(P\) is a \(de\)-representation of \(G\circ H\), the definition
of the fractional Haemers number gives
\[
\mathcal{H}_f(G\circ H;\mathbb{F})
\leq \frac{\operatorname{rank}_{\mathbb{F}}(P)}{de} 
= \frac{\operatorname{rank}_{\mathbb{F}}(M)}{d}\, \frac{\operatorname{rank}_{\mathbb{F}}(N)}{e}.
\]
This inequality holds for every \(d\)-representation \(M\) of \(G\) and
every \(e\)-representation \(N\) of \(H\) over \(\mathbb{F}\). Since
these choices are independent, taking the infimum of the rightmost term
over all integers \(d,e \geq 1\) and all corresponding representations
\(M\) and \(N\), we obtain
\begin{align}
\label{eq7: 06.08.26}
\mathcal{H}_f(G\circ H;\mathbb{F})
\leq \mathcal{H}_f(G;\mathbb{F}) \, \mathcal{H}_f(H;\mathbb{F}).
\end{align}

\smallskip \noindent {\textbf{Reverse inequality.}}
We next prove the reverse inequality. 
The argument consists of four steps. Starting from an arbitrary \(k\)-representation \(P\) 
of \(G\circ H\), we first restrict \(P\) to the copies of \(H\). We then compress these 
restrictions to a common dimension \(r\) and use them to construct an \(r\)-representation 
of \(G\). Finally, we combine the resulting bounds for \(G\) and \(H\).

\smallskip \noindent {\textbf{Step 1: Restriction to the copies of \( H \).}}
Let \(k\) be an arbitrary positive 
integer and let \(P\) be a \(k\)-representation of \(G\circ H\) over
\(\mathbb{F}\). Set
\begin{align}
\label{eq8: 06.08.26}
n\coloneqq\operatorname{rank}_{\mathbb{F}}(P),
\end{align}
and choose a rank factorization
\begin{align}
\label{eq9: 06.08.26}
P=AB,
\end{align}
where \(A\) has \(n\) columns and \(B\) has \(n\) rows.

\noindent 
The matrix \(P\) is a block matrix indexed by \(V(G)\times V(H)\),
with each block having size \(k\times k\). Thus, its scalar rows and
columns may be indexed by
\[
V(G)\times V(H) \times [k],
\qquad
[k] \coloneqq \{1,\ldots,k\}.
\]
The rows of \(A\) and the columns of \(B\) inherit the corresponding
scalar indices from the rank factorization \(P=AB\).

\noindent 
For each \(u\in V(G)\), let \(A_u\) be the submatrix of \(A\) whose
rows are indexed by
\[
\{u\}\times V(H)\times[k],
\]
and let \(B_u\) be the submatrix of \(B\) whose columns are indexed by
the same set. Then
\begin{align}
\label{eq10: 06.08.26}
P_u \coloneqq A_u B_u
\end{align}
is the principal submatrix of \(P\) corresponding to
\(\{u\}\times V(H)\). In particular, \(P_u\) is a
\(k\)-representation of \(H\). Define
\begin{align}
\label{eq11: 06.08.26}
r_u\coloneqq\operatorname{rank}_{\mathbb{F}}(P_u),
\qquad
r\coloneqq\min_{u\in V(G)}r_u.
\end{align}
Since \(P_u\) contains an identity matrix \(I_k\) as a diagonal
block, we have \(r_u\geq k\) for every \(u\in V(G)\). In particular,
by \eqref{eq11: 06.08.26},
\[
r \geq k > 0.
\]
Since \(P_u\) is a \(k\)-representation of \(H\) over $\mathbb{F}$, 
from \eqref{eq: fractional Haemers number} and \eqref{eq11: 06.08.26},
for every \(u\in V(G)\) we have
\[
\frac{r_u}{k} \geq \mathcal{H}_f(H;\mathbb{F}).
\]
Taking the minimum over \(u\in V(G)\), we obtain
\begin{align}
\label{eq12: 06.08.26}
\frac{r}{k} \geq \mathcal{H}_f(H;\mathbb{F}).
\end{align}

\smallskip \noindent {\textbf{Step 2: Compression to the common dimension \(r\).}}
We next prepare the construction of an \(r\)-representation of \(G\). For each
\(u\in V(G)\), we have by \eqref{eq10: 06.08.26} and \eqref{eq11: 06.08.26}
\begin{align}
\label{eq13: 06.08.26}
\operatorname{rank}_{\mathbb{F}}(A_u B_u) = r_u \geq r.
\end{align}
Consequently, \(A_u B_u\) contains a nonsingular \(r \times r\) submatrix. 
Hence there exist binary row- and column-selection matrices
\[
R_u \in \mathbb{F}^{\, r\times k|V(H)|}, \qquad S_u \in \mathbb{F}^{\, k|V(H)|\times r},
\]
respectively, such that
\begin{align}
\label{eq14: 06.08.26}
T_u \coloneqq R_u A_u B_u S_u
\end{align}
is nonsingular. Define
\begin{align}
\label{eq15: 06.08.26}
X_u \coloneqq T_u^{-1} R_u, \qquad Y_u \coloneqq S_u.
\end{align}
Then, by \eqref{eq14: 06.08.26} and \eqref{eq15: 06.08.26},
\begin{align}
\label{eq16: 06.08.26}
X_u A_u B_u Y_u
= T_u^{-1} R_u A_u B_u S_u
= T_u^{-1} T_u
= I_r.
\end{align}
Upon setting
\begin{align}
\label{eq17: 06.08.26}
C_u \coloneqq X_u A_u \in \mathbb{F}^{r\times n}, \qquad
D_u \coloneqq B_u Y_u \in\mathbb{F}^{n\times r},
\end{align}
it follows from \eqref{eq16: 06.08.26} and \eqref{eq17: 06.08.26} that
\begin{align}
\label{eq18: 06.08.26}
C_u D_u = X_u A_u B_u Y_u = I_r.
\end{align}

\smallskip \noindent {\textbf{Step 3: Construction of an \(r\)-representation of \(G\).}}
Define a block matrix \(Q\), indexed by \(V(G)\), by
\begin{align}
\label{eq19: 06.08.26}
Q_{u,v} \coloneqq C_u D_v.
\end{align}
By \eqref{eq18: 06.08.26}, its diagonal blocks satisfy
\begin{align}
\label{eq20: 06.08.26}
Q_{u,u} = C_u D_u = I_r.
\end{align}
Now suppose that \(u\neq v\) and \(\{u,v\}\notin E(G)\). Every vertex
in the copy of \(H\) corresponding to \(u\) is then nonadjacent in
\(G\circ H\) to every vertex in the copy corresponding to \(v\).
Hence the corresponding submatrix of \(P\) is zero, that is,
\begin{align}
\label{eq21: 06.08.26}
A_u B_v = O_{k \, |V(H)|},
\end{align}
and since \(C_u = X_u A_u\) and \(D_v = B_v Y_v\) (see \eqref{eq17: 06.08.26}), we obtain
\begin{align}
\label{eq22: 06.08.26}
C_u D_v = X_u A_u B_v Y_v = O_r.
\end{align}
Thus, \(Q\) is an \(r\)-representation of \(G\) over \( \mathbb{F} \).

It remains to relate the rank of \(Q\) to that of \(P\).
Fix an ordering \(V(G)=\{u_1,\ldots,u_m\}\), and define
\begin{align}
\label{eq23: 06.08.26}
\widehat C\coloneqq
\begin{pmatrix}
C_{u_1}\\
\vdots\\
C_{u_m}
\end{pmatrix},
\qquad
\widehat D\coloneqq
\begin{pmatrix}
D_{u_1} & \cdots & D_{u_m}
\end{pmatrix}.
\end{align}
By the definition of the blocks of \(Q\),
\begin{align}
\label{eq24: 06.08.26}
Q=\widehat C\widehat D.
\end{align}
Since \(\widehat C\) has \(n\) columns and \(\widehat D\) has \(n\)
rows, it follows that
\begin{align}
\label{eq25: 06.08.26}
\operatorname{rank}_{\mathbb{F}}(Q)\leq n,
\end{align}
and therefore
\begin{align}
\mathcal{H}_f(G;\mathbb{F})
& \leq \frac{\operatorname{rank}_{\mathbb{F}}(Q)}{r} \nonumber \\
\label{eq26: 06.08.26}
& \leq \frac{n}{r}.
\end{align}

\smallskip \noindent {\textbf{Step 4: Completion of the proof.}}
Combining \eqref{eq8: 06.08.26}, \eqref{eq12: 06.08.26}, and \eqref{eq26: 06.08.26} yields 
\begin{align}
\mathcal{H}_f(G;\mathbb{F}) \, \mathcal{H}_f(H;\mathbb{F})
&\leq \frac{n}{r}\,\frac{r}{k} \nonumber \\
\label{eq27: 06.08.26}
&= \frac{\operatorname{rank}_{\mathbb{F}}(P)}{k}.
\end{align}
Since this holds for all integers \(k \geq 1\) and all \(k\)-representations 
\(P\) of \(G\circ H\), taking the infimum of the rightmost term in 
\eqref{eq27: 06.08.26} over all such \(k\) and \(P\) yields
\begin{align}
\label{eq28: 06.08.26}
\mathcal{H}_f(G;\mathbb{F})\, \mathcal{H}_f(H;\mathbb{F})
\leq \mathcal{H}_f(G\circ H;\mathbb{F}).
\end{align}
Combining inequalities \eqref{eq7: 06.08.26} and \eqref{eq28: 06.08.26} proves 
equality \eqref{eq3: 05.08.26}.
\end{proof}

\section{Shannon capacity for lexicographic graph products}
\label{section: Shannon capacity under a lexicographic product of graphs}

Our first result establishes a general comparison among the Shannon
capacities of the two ordered lexicographic products, the strong product, and
the individual factors. Since the lexicographic product is not commutative in
general, both \(G\circ H\) and \(H\circ G\) must be considered. The strong
product \(G\boxtimes H\) is a spanning subgraph of \(G\circ H\), and, up to
the natural interchange of coordinates, also of \(H\circ G\). Hence, the
monotonicity of Shannon capacity with respect to the addition of edges gives
an upper bound on the capacities of both lexicographic products. Together
with a product construction for independent sets and the multiplicativity of
the Lov\'{a}sz theta function under the strong product, this yields the
following chain of bounds.

\begin{theorem}[Bounds for the Shannon capacity]
\label{theorem: capacity bounds of a lexicographic product}
For all finite simple graphs \(G\) and \(H\),
\begin{align}
\Theta(G) \, \Theta(H) &\leq \min\{ \Theta(G \circ H), \Theta(H \circ G) \} \nonumber \\
& \leq \max\{ \Theta(G \circ H), \Theta(H \circ G) \} \nonumber \\
\label{eq: capacity bounds of a lexicographic product}
& \leq \Theta(G \boxtimes H) \leq \vartheta(G) \, \vartheta(H).
\end{align}
In particular,
\begin{align}
\label{eq39: 06.08.26}
\hspace*{-0.2cm} \Theta(G \boxtimes H)=\Theta(G)\,\Theta(H)
&\;\Longrightarrow\;
\Theta(G \circ H)=\Theta(G)\,\Theta(H), \\
\hspace*{-0.2cm} \bigl(\Theta(G)=\vartheta(G),\;
\Theta(H)=\vartheta(H)\bigr)
&\;\Longrightarrow\;
\Theta(G \circ H)=\Theta(H \circ G) \notag\\
&\hphantom{\;\Longrightarrow\;\Theta(G \circ H)}
\kern0.1cm
=\vartheta(G)\,\vartheta(H).
\label{eq40: 06.08.26}
\end{align}
\end{theorem}

\begin{proof}
Since the strong product is commutative, whereas the lexicographic
product is generally noncommutative (see
Remark~\ref{remark: commutativity}), it suffices to prove that
\begin{align}
\Theta(G)\,\Theta(H)
\leq \Theta(G\circ H)
\leq \Theta(G\boxtimes H)
\leq \vartheta(G)\,\vartheta(H).
\end{align}

\noindent 
Fix \(n \geq 1\), and let
\begin{align}
\label{eq29: 06.08.26}
\mathcal{A} \subseteq V\!\left(G^{\boxtimes n}\right), 
\quad \mathcal{B} \subseteq V\!\left(H^{\boxtimes n}\right)
\end{align}
be maximum independent sets. Define the map
\begin{align}
\label{eq30: 06.08.26}
\Phi_n\colon
V\!\left(G^{\boxtimes n}\right)
\times
V\!\left(H^{\boxtimes n}\right)
\longrightarrow
V\!\left((G\circ H)^{\boxtimes n}\right)
\end{align}
by
\begin{align}
\label{eq31: 06.08.26}
\Phi_n(\mathbf g,\mathbf h) \coloneqq \bigl((g_1,h_1),\ldots,(g_n,h_n)\bigr),
\end{align}
where 
\[
\mathbf g=(g_1,\ldots,g_n), \qquad \mathbf h=(h_1,\ldots,h_n).
\]
We show that
\begin{align}
\label{eq33: 06.08.26}
\mathcal{C} \coloneqq \Phi_n(\mathcal{A} \times \mathcal{B})
\end{align}
is an independent set in \((G\circ H)^{\boxtimes n}\).
Indeed, consider two distinct elements
\(\Phi_n(\mathbf g,\mathbf h)\) and
\(\Phi_n(\mathbf g',\mathbf h')\) of \(\mathcal C\).
If \(\mathbf g\neq\mathbf g'\), then the independence of
\(\mathcal A\) implies that there exists a coordinate \(j\) such that
\(g_j\neq g'_j \) and \(g_j\not\sim_G g'_j\).
It follows from the definition of the lexicographic product that
\((g_j,h_j)\) and \((g'_j,h'_j)\) are nonadjacent in \(G\circ H\).
Otherwise, if \(\mathbf g=\mathbf g'\), then \(\mathbf h\neq\mathbf h'\).
The independence of \(\mathcal B\) therefore implies that there exists
a coordinate \(j\) such that \(h_j\neq h'_j\) and \(h_j\not\sim_H h'_j\).
Since \(g_j=g'_j\), the vertices \((g_j,h_j)\) and
\((g'_j,h'_j)\) are again nonadjacent in \(G\circ H\).
Thus, \(\mathcal C\) is independent.
Furthermore, since \(\Phi_n\) is injective,
\begin{align}
\label{eq34: 06.08.26}
|\mathcal{C}| = |\mathcal{A} \times \mathcal{B}| = |\mathcal{A}|\,|\mathcal{B}|.
\end{align}
Since \(\mathcal{A}\) and \(\mathcal{B}\) are, by assumption, maximum
independent sets in \(G^{\boxtimes n}\) and \(H^{\boxtimes n}\),
respectively, and the set \(\mathcal{C}\) defined in \eqref{eq33: 06.08.26} 
is independent in \((G\circ H)^{\boxtimes n}\) with cardinality given by
\eqref{eq34: 06.08.26}, it follows that
\begin{align}
\label{eq35: 06.08.26}
\alpha\!\left((G\circ H)^{\boxtimes n}\right)
\geq \alpha\!\left(G^{\boxtimes n}\right) \,
\alpha\!\left(H^{\boxtimes n}\right).
\end{align}
Taking \(n\)-th roots of both sides of \eqref{eq35: 06.08.26} and letting 
\(n\to\infty\) yields
\begin{align}
\label{eq36: 06.08.26}
\Theta(G\circ H)\geq\Theta(G)\,\Theta(H).
\end{align}

\smallskip 
\noindent 
By Lemma~\ref{lemma: spanning subgraph}, \(G\boxtimes H\) is a
spanning subgraph of \(G\circ H\). Since Shannon capacity is
nonincreasing under the addition of edges, it follows that
\begin{align}
\label{eq37: 06.08.26}
\Theta(G\circ H)\leq\Theta(G\boxtimes H).
\end{align}
Finally, by \cite[Theorems~1 and~7]{Lovasz79},
\begin{align}
\label{eq38: 06.08.26}
\Theta(G\boxtimes H)
\leq \vartheta(G\boxtimes H)
= \vartheta(G)\,\vartheta(H).
\end{align}
Combining \eqref{eq36: 06.08.26}, \eqref{eq37: 06.08.26}, and 
\eqref{eq38: 06.08.26} completes the proof of 
\eqref{eq: capacity bounds of a lexicographic product}. 
Finally, assertions \eqref{eq39: 06.08.26} and \eqref{eq40: 06.08.26}
follow directly from \eqref{eq: capacity bounds of a lexicographic product}. 
\end{proof}

\begin{corollary}[Shannon capacity of the lexicographic power of a graph]
\label{corollary: Shannon capacity of the lexicographic power of G}
Let \(G\) be a finite, undirected, and simple graph, and let \( m \geq 1 \) be an integer. Then, 
\begin{align}
\label{eq41: 06.08.26}
\Theta(G^{\, \circ m})=\Theta(G)^m,
\end{align}
where \( G^{\, \circ m} \) denotes the $m$-fold lexicographic power of \( G \).  
\end{corollary}

\begin{proof}
By Theorem~\ref{theorem: capacity bounds of a lexicographic product}, it follows that 
for every integer \( m \geq 1 \) 
\begin{align*}
\Theta(G)^m & \leq \Theta(G^{\, \circ m}) \\
& \leq \Theta(G^{\, \boxtimes m})  \\
&=\Theta(G)^m,
\end{align*}
thus proving \eqref{eq41: 06.08.26}.
\end{proof}

\noindent
Theorem~\ref{theorem: capacity bounds of a lexicographic product} states that
\[
\min\{\Theta(G\circ H), \, \Theta(H\circ G)\} \geq \Theta(G) \, \Theta(H).
\]
The next result, Proposition~\ref{proposition: independence numbers of squared lexicographic products}, 
gives a general graph--complement construction that is subsequently used in 
Proposition~\ref{proposition: explicit graphs satisfying strict graph complement condition} to establish 
strict inequalities for several strongly regular graphs.
Three countably infinite families of graph--complement pairs are
subsequently obtained in
Theorem~\ref{theorem: strict supermultiplicativity of capacity}
by taking lexicographic powers of these graphs; the strict inequality
holds throughout these families, and the multiplicative gap is
unbounded within each one.

\begin{proposition}
\label{proposition: independence numbers of squared lexicographic products}
Let \(G\) be a finite simple graph on \(n\) vertices. Then
\begin{align}
\label{eq: general lower bound G circ complement}
\alpha\bigl((G \circ \overline{G})^{\boxtimes \, 2}\bigr)
&\geq n\,\alpha(G)\,\omega(G),\\
\label{eq: general lower bound complement circ G}
\alpha\bigl((\overline{G} \circ G)^{\boxtimes \, 2}\bigr)
&\geq n\,\alpha(G)\,\omega(G).
\end{align}
Consequently,
\begin{align}
\min\bigl\{\Theta(G \circ \overline{G}), \, \Theta(\overline{G} \circ G)\bigr\}
\geq \sqrt{n\,\alpha(G)\,\omega(G)}.
\label{eq: general capacity lower bound for graph complement products}
\end{align}
In particular, if
\begin{align}
\label{eq: sufficient condition for strict supermultiplicativity}
\sqrt{n\,\alpha(G)\,\omega(G)}
>\Theta(G)\,\Theta(\overline{G}),
\end{align}
then
\begin{align}
\label{eq: strict supermultiplicativity from graph complement construction}
\min\bigl\{\Theta(G \circ \overline{G}), \, \Theta(\overline{G} \circ G)\bigr\}
>\Theta(G)\,\Theta(\overline{G}).
\end{align}
\end{proposition}

\begin{proof}
Let \(V\) be the common vertex set of \(G\) and \(\overline G\), with
\(\lvert V\rvert=n\). Choose maximum independent sets
\[
A\subseteq V(G),\qquad B\subseteq V(\overline G).
\]
Then \(\lvert A\rvert=\alpha(G)\) and
\(\lvert B\rvert=\alpha(\overline G)=\omega(G)\).

\smallskip 
\noindent 
Consider the vertex subset of \( (G \circ \overline{G})^{\boxtimes \, 2} \) given by 
\begin{align}
\label{eq: general independent set I}
\mathcal{I} \coloneqq \left\{ \bigl((a,v),(v,b)\bigr): a\in A,\ v\in V,\ b\in B \right\},
\end{align}
whose cardinality is 
\[
\lvert \mathcal{I} \rvert = n\,\alpha(G)\,\omega(G).
\]
We claim that \(\mathcal{I}\) is an independent set in
\((G \circ \overline{G})^{\boxtimes \, 2}\). Let
\[
x=\bigl((a,v),(v,b)\bigr), \qquad x'=\bigl((a',v'),(v',b')\bigr)
\]
be distinct elements of \(\mathcal I\). If \(a\neq a'\), then 
\(a \not\sim_G a'\), since \(A\) is an independent set in \( G \).
Thus, \((a,v)\) and \((a',v')\) are distinct and nonadjacent in
\(G \circ \overline{G}\). If \(a=a'\) and \(v\neq v'\), then
\begin{align}
(a,v) \sim_{G \circ \overline{G}} (a,v')
&\quad \Longleftrightarrow \quad v \sim_{\overline{G}} v',\\
(v,b) \sim_{G \circ \overline{G}} (v',b')
& \quad \Longleftrightarrow \quad v \sim_G v'.
\end{align}
The two relations on the right cannot hold simultaneously. Finally, if
\(a=a'\) and \(v=v'\), then \(b \neq b'\), and the independence of \(B\)
in \(\overline{G}\) implies that \((v,b)\) and \((v,b')\) are distinct and
nonadjacent in \(G \circ \overline{G}\). Thus, in every case, \(x\) and \(x'\)
are nonadjacent in the strong product. This proves
\eqref{eq: general lower bound G circ complement}.

\smallskip
\noindent 
For the reversed order, consider the vertex subset of
\((\overline{G} \circ G)^{\boxtimes 2}\) given by
\begin{align}
\label{eq: general independent set J}
\mathcal{J} \coloneqq \left\{ \bigl((b,v),(v,a)\bigr): b\in B,\ v\in V,\ a\in A \right\},
\end{align}
whose cardinality is
\[
\lvert\mathcal{J}\rvert=n\,\alpha(G)\,\omega(G).
\]
The same argument, with \(G\) and \(\overline{G}\) interchanged, shows that
\(\mathcal{J}\) is an independent set in \((\overline{G} \circ G)^{\boxtimes \, 2}\).
This proves \eqref{eq: general lower bound complement circ G}.

\smallskip
\noindent 
By the definition of the Shannon capacity,
\[
\Theta(F)\geq\sqrt{\alpha(F^{\boxtimes \, 2})}
\]
for every finite graph \(F\). Applying this inequality to
\(G \circ \overline{G}\) and \(\overline{G} \circ G\) proves
\eqref{eq: general capacity lower bound for graph complement products}.

\smallskip
\noindent 
Combining \eqref{eq: general capacity lower bound for graph complement products}
with condition~\eqref{eq: sufficient condition for strict supermultiplicativity}
immediately yields
\eqref{eq: strict supermultiplicativity from graph complement construction}.
\end{proof}

\begin{proposition}
\label{proposition: explicit graphs satisfying strict graph complement condition}
Condition~\eqref{eq: sufficient condition for strict supermultiplicativity}
holds for each of the following graphs and their complements:
\begin{enumerate}
\item The Schl\"{a}fli graph \(S=\operatorname{srg}(27,16,10,8)\);
\item The McLaughlin graph \(M=\operatorname{srg}(275,112,30,56)\);
\item The second subconstituent of the McLaughlin graph 
\[
L=\operatorname{srg}(162,56,10,24).
\]
\end{enumerate}
\end{proposition}

\begin{proof}
Condition~\eqref{eq: sufficient condition for strict supermultiplicativity} is verified 
for each of the three strongly regular graphs listed above.
\begin{enumerate}
\item 
For the Schl\"{a}fli graph (see \cite[Section~10.10]{BrouwerVanMaldeghem2022}),
\[
\lvert V(S)\rvert=27, \qquad \alpha(S)=3, \qquad \omega(S)=6,
\]
and
\[
\Theta(S)=3, \qquad \Theta(\overline{S}) \leq 7,
\]
where the first equality is given in \cite[Example~16]{Sason23} and the latter inequality 
follows from Haemers' bound~\cite{Haemers79}.
Hence
\begin{align*}
\Theta(S) \, \Theta(\overline{S}) \leq 21 
< \sqrt{\lvert V(S)\rvert \, \alpha(S) \, \omega(S)}
= 9 \sqrt{6} \approx 22.045.
\end{align*}

\item 
The McLaughlin graph (see \cite[Section~10.61]{BrouwerVanMaldeghem2022}) has
\[
\lvert V(M)\rvert = 275, \qquad \alpha(M) = 22, \qquad \omega(M) = 5,
\]
and adjacency spectrum
\[
112^{(1)}, \qquad 2^{(252)}, \qquad (-28)^{(22)}.
\]
Let \(A_M\) denote the adjacency matrix of \(M\). The matrix
\( I-\frac{1}{2}A_M \) fits \(M\) over \(\mathbb{R}\): its 
diagonal entries are all equal to \(1\), while, for every pair 
of distinct nonadjacent vertices \(u,v\in V(M)\),
\[
\left(I-\frac{1}{2}A_M\right)_{u,v}
=-\frac{1}{2}(A_M)_{u,v}=0.
\]
Moreover, multiplication of this matrix by the nonzero scalar \(-2\) does not
change the rank. Therefore, Haemers' rank bound~\cite{Haemers79} gives
\[
\Theta(M)\leq\operatorname{rank}_{\mathbb{R}}(A_M-2I).
\]
Since \(2\) is an eigenvalue of \(A_M\) with multiplicity \(252\), the
nullity of \(A_M-2I\) is \(252\). Consequently,
\[
\Theta(M)
\leq\operatorname{rank}_{\mathbb{R}}(A_M-2I)
=275-252
=23.
\]
The complement \(\overline{M}\) is \(162\)-regular and has least adjacency
eigenvalue \(-3\). For strongly regular graphs, the ratio bound for the Lov\'{a}sz
theta function is attained (see \cite[Proposition~1]{Sason23}), which yields 
\[
\Theta(\overline{M}) \leq \vartheta(\overline{M})
=\frac{275\cdot3}{162+3} = 5.
\]
Since also \(\alpha(\overline{M}) = \omega(M) = 5\), the equality  
\( \Theta(\overline{M}) = 5 \) follows, and consequently  
\begin{align*}
\Theta(M) \, \Theta(\overline{M})
& \leq 23 \cdot 5 = 115 \\
& < \sqrt{\lvert V(M)\rvert \, \alpha(M) \, \omega(M)}
=55 \sqrt{10} \approx 173.925.
\end{align*}

\item
Finally, let \(L\) be the second subconstituent of the McLaughlin graph
(see \cite[Section~10.48]{BrouwerVanMaldeghem2022}). It satisfies
\[
\lvert V(L)\rvert=162, \qquad \alpha(L)=21, \qquad \omega(L)=3,
\]
and has adjacency spectrum
\[
56^{(1)}, \qquad 2^{(140)}, \qquad (-16)^{(21)}.
\]
Let \(A_L\) denote its adjacency matrix. The matrix \(A_L-2I\) fits
\(L\) over \(\mathbb{R}\): its diagonal entries are all equal to \(-2\)
and hence are nonzero, whereas its entries corresponding to distinct
nonadjacent vertices of \(L\) are zero. Therefore, Haemers' rank
bound~\cite{Haemers79} gives
\[
\Theta(L) \leq \operatorname{rank}_{\mathbb{R}}(A_L-2I).
\]
Since \(2\) is an eigenvalue of \(A_L\) with multiplicity \(140\), the
nullity of \(A_L-2I\) is \(140\). Consequently,
\[
\Theta(L) \leq \operatorname{rank}_{\mathbb{R}}(A_L-2I) = 162-140 = 22.
\]
The complement \(\overline{L}\) is \(105\)-regular and has least
adjacency eigenvalue \(-3\). 
For strongly regular graphs, the ratio bound for the Lov\'{a}sz
theta function is attained (see \cite[Proposition~1]{Sason23}), 
which yields 
\[
\Theta(\overline{L})
\leq \vartheta(\overline{L})
=\frac{162 \cdot 3}{105+3}
=\frac{9}{2}.
\]
It follows that
\begin{align*}
\Theta(L) \, \Theta(\overline{L})
& \leq 22 \cdot \tfrac{9}{2} = 99 \\
& <\sqrt{\lvert V(L)\rvert \, \alpha(L) \, \omega(L)}
= 27\sqrt{14}
\approx 101.025.
\end{align*}
\end{enumerate}

Condition~\eqref{eq: sufficient condition for strict supermultiplicativity}
is invariant under complementation, since complementation interchanges
\(\alpha(G)\) and \(\omega(G)\) and leaves \(\Theta(G)\,\Theta(\overline G)\) 
unchanged. The complements of the three listed graphs therefore satisfy the 
condition as well.
\end{proof}

\begin{corollary}
\label{corollary: strict supermultiplicativity and superadditivity for explicit graphs}
Let \(G\) be the Schl\"{a}fli graph, the McLaughlin graph, or the second
subconstituent of the McLaughlin graph. Then
\begin{align}
\label{eq: strict lexicographic supermultiplicativity for explicit graphs}
\min\bigl\{
\Theta(G\circ\overline{G}),
\Theta(\overline{G}\circ G)
\bigr\}
>
\Theta(G)\,\Theta(\overline{G}).
\end{align}
Consequently, 
\begin{align}
\label{eq: strict strong product for explicit graphs}
\Theta(G \boxtimes \overline{G}) &> \Theta(G) \, \Theta(\overline{G}),  \\
\label{eq: strict disjoint union for explicit graphs}
\Theta(G+\overline{G}) &> \Theta(G) + \Theta(\overline{G}).
\end{align}
\end{corollary}

\begin{proof}
The first assertion follows immediately from
Propositions~\ref{proposition: independence numbers of squared lexicographic products}
and~\ref{proposition: explicit graphs satisfying strict graph complement condition}.
More explicitly, the bounds obtained for the three graphs are
\begin{align}
\min \bigl\{ \Theta(S \circ \overline{S}), \, \Theta(\overline{S} \circ S) \bigr\}
&\geq 9 \sqrt{6} > 21 \geq \Theta(S) \, \Theta(\overline{S}),\\
\min\bigl\{ \Theta(M \circ \overline{M}), \, \Theta(\overline{M}\circ M) \bigr\}
&\geq 55 \sqrt{10} > 115 \geq \Theta(M) \, \Theta(\overline{M}),\\
\min\bigl\{ \Theta(L \circ \overline{L}), \, \Theta(\overline{L}\circ L) \bigr\}
&\geq 27 \sqrt{14} > 99 \geq \Theta(L) \, \Theta(\overline{L}).
\end{align}

For every pair of finite simple graphs \(F\) and \(H\), the strong
product \(F\boxtimes H\) is a spanning subgraph of the lexicographic
product \(F\circ H\); see Lemma~\ref{lemma: spanning subgraph}.
Since Shannon capacity is nonincreasing under the addition of edges,
\[
\Theta(F \circ H) \leq \Theta(F \boxtimes H).
\]
Applying this inequality with \(F=G\) and \(H=\overline{G}\), and using
\eqref{eq: strict lexicographic supermultiplicativity for explicit graphs},
gives \eqref{eq: strict strong product for explicit graphs}.

Finally, the duality theorem established in \cite{Schrijver23, WigdersonZ26} 
states that, for all finite simple graphs \(F\) and \(H\),
\begin{equation}
\label{eq:shannon_union_equivalent_product}
\Theta(F+H) = \Theta(F) + \Theta(H)
\quad \Longleftrightarrow \quad
\Theta(F \boxtimes H) = \Theta(F) \, \Theta(H).
\end{equation}
Taking \(F=G\) and \(H=\overline{G}\), inequality
\eqref{eq: strict strong product for explicit graphs} implies that
\[
\Theta(G+\overline{G}) \neq \Theta(G) + \Theta(\overline{G}).
\]
Since the Shannon capacity is superadditive with respect to disjoint
unions, the last inequality proves \eqref{eq: strict disjoint union for explicit graphs}.
\end{proof}

\begin{remark}
{\em
The Schl\"{a}fli graph recovers Alon's example~\cite{Alon98}, whereas
the McLaughlin graph and its second subconstituent provide two
additional explicit graph--complement pairs for which Shannon capacity
is strictly superadditive under disjoint union. Alon's example
disproved Shannon's conjecture~\cite{Shannon56} that the Shannon
capacity of a disjoint union is always equal to the sum of the
capacities of its components.}
\end{remark}

\smallskip 
Building on
Propositions~\ref{proposition: independence numbers of squared lexicographic products}
and~\ref{proposition: explicit graphs satisfying strict graph complement condition},
we construct three countably infinite families of graphs satisfying
\begin{align}
\label{eq1: 23.08.26}
\min\bigl\{\Theta(G \circ \overline{G}), \, \Theta(\overline{G}\circ G)\bigr\}
> \Theta(G) \, \Theta(\overline{G}).
\end{align}
Moreover, each family exhibits an arbitrarily large multiplicative
gap: for every \(C>1\), it contains a graph \(G\) for which the ratio
of the left-hand side of \eqref{eq1: 23.08.26} to its right-hand side
is at least \(C\).

\begin{theorem}
\label{theorem: strict supermultiplicativity of capacity}
Let \(S\), \(M\), and \(L\) denote, respectively, the Schl\"{a}fli
graph, the McLaughlin graph, and the second subconstituent of the
McLaughlin graph. For \( X\in\{S,M,L\} \) and an integer \(m \geq 1\), let
\begin{align}
\label{eq: X_m}
X_m \coloneqq X^{\circ m}
\end{align}
be the \(m\)-fold lexicographic power of \(X\). Define
\begin{align}
& \mathcal{G}_X \coloneqq \{X_m: m \geq 1\},
\qquad
\mathcal{G} \coloneqq
\mathcal{G}_S \cup \mathcal{G}_M \cup \mathcal{G}_L, \\[0.1cm]
\label{eq: R}
& R_X(m) \coloneqq
      \frac{\min\bigl\{\Theta(X_m \circ \overline{X_m}),\,
      \Theta(\overline{X_m} \circ X_m)\bigr\}}
      {\Theta(X_m) \, \Theta(\overline{X_m})}.
\end{align}      
Then, for all \( m \geq 1 \),
\begin{align}
R_S(m) &\geq
\left(\frac{54}{49}\right)^{m/2},
\label{eq: gap for Schlaefli powers}\\[0.1cm]
R_M(m) &\geq
\left(\frac{1210}{529}\right)^{m/2},
\label{eq: gap for McLaughlin powers}\\[0.1cm]
R_L(m) &\geq
\left(\frac{126}{121}\right)^{m/2}.
\label{eq: gap for second subconstituent powers}
\end{align}
The right-hand side of each of these three inequalities
is strictly larger than \(1\) for every
\(m \geq 1\) and tends to infinity as \(m\to\infty\).
Consequently, every graph in \(\mathcal{G}\) satisfies
\eqref{eq1: 23.08.26}, and each of the three countably infinite
families \(\mathcal{G}_S\), \(\mathcal{G}_M\), and
\(\mathcal{G}_L\) exhibits an arbitrarily large multiplicative gap.
\end{theorem}

\begin{proof}
Fix \(X\in\{S,M,L\}\) and an integer \(m \geq 1\). By the
multiplicative formulas for the order, independence number, and clique number
under lexicographic products (see
Proposition~\ref{proposition: independence and clique numbers of lexicographic product}),
\begin{equation}
\label{eq: formulas for X_m}
\begin{aligned}
\lvert V(X_m)\rvert &= \lvert V(X)\rvert^m, \\
\alpha(X_m) &= \alpha(X)^m,  \\
\omega(X_m) &= \omega(X)^m.
\end{aligned}
\end{equation}
Applying Proposition~\ref{proposition: independence numbers of squared lexicographic products}
to \(X_m\) gives
\begin{align}
& \hspace*{-0.4cm} \min\bigl\{ \Theta(X_m\circ\overline{X_m}), \, \Theta(\overline{X_m}\circ X_m) \bigr\} \notag \\
& \geq \sqrt{\lvert V(X_m)\rvert \, \alpha(X_m) \, \omega(X_m)} \notag\\
& = \bigl(\lvert V(X)\rvert \, \alpha(X) \, \omega(X)\bigr)^{m/2}.
\label{eq: general lower bound for powers of explicit graphs}
\end{align}
Furthermore, complementation commutes with the lexicographic product
(see \eqref{eq: complement of lexicographic product}), so
\[
\overline{X_m} = \overline{X^{\circ m}} = (\overline{X})^{\circ m}.
\]
Hence, by Corollary~\ref{corollary: Shannon capacity of the lexicographic power of G},
\[
\Theta(X_m)=\Theta(X)^m, \qquad \Theta(\overline{X_m})=\Theta(\overline{X})^m.
\]
Combining these identities with \eqref{eq: R} and 
\eqref{eq: general lower bound for powers of explicit graphs} yields
\begin{align}
R_X(m) \geq 
\left( \frac{\sqrt{\lvert V(X)\rvert \, \alpha(X) \, \omega(X)}}{\Theta(X) \, \Theta(\overline{X})}\right)^m.
\end{align}
Using the graph parameters and the upper bounds on the relevant
Shannon capacities established in the proof of
Proposition~\ref{proposition: explicit graphs satisfying strict graph complement condition},
we obtain, for all integers \(m \geq 1\),
\begin{align*}
R_S(m) &\geq \left( \frac{\sqrt{27 \cdot 3 \cdot 6}}{3 \cdot 7} \right)^m = \left( \frac{54}{49} \right)^{m/2},\\[0.1cm]
R_M(m) &\geq \left( \frac{\sqrt{275 \cdot 22 \cdot 5}}{23 \cdot 5} \right)^m = \left( \frac{1210}{529} \right)^{m/2},\\[0.1cm]
R_L(m) &\geq \left( \frac{\sqrt{162 \cdot 21 \cdot 3}}{22 \cdot \frac92} \right)^m = \left( \frac{126}{121} \right)^{m/2}.
\end{align*}
These are precisely \eqref{eq: gap for Schlaefli powers}--\eqref{eq: gap for second subconstituent powers}.
Since \( \frac{54}{49}>1, \,  \frac{1210}{529}>1 \), and 
\(\frac{126}{121}>1\), the right-hand side of each inequality in 
\eqref{eq: gap for Schlaefli powers}--\eqref{eq: gap for second subconstituent powers}
is strictly larger than \(1\) for every integer \(m \geq 1\) and tends to infinity as \(m\to\infty\).
Therefore, every graph \(G \in \mathcal{G}\) satisfies \eqref{eq1: 23.08.26}.
More precisely, define 
\[
\rho_S \coloneqq \frac{54}{49} \approx 1.10204, \qquad
\rho_M \coloneqq \frac{1210}{529} \approx 2.28733, \qquad
\rho_L \coloneqq \frac{126}{121} \approx 1.04132.
\]
Given \(C>1\) and \(X\in\{S,M,L\}\), choosing
\[
m = \left\lceil \frac{2\ln C}{\ln\rho_X} \right\rceil
\]
ensures that \( R_X(m) \geq C\). Thus, for each of the families
\(\mathcal{G}_S\), \(\mathcal{G}_M\), and \(\mathcal{G}_L\),
the strict supermultiplicativity inequality
\eqref{eq1: 23.08.26} holds with an arbitrarily large multiplicative
gap. Finally, each family is countably infinite. 
Indeed, for every \(X\in\{S,M,L\}\), the family of graphs \(\mathcal{G}_X\) 
is indexed by the integers \(m \geq 1\), and
\( \lvert V(X_m)\rvert = \lvert V(X)\rvert^m \), 
which implies that \(X_m\) and \(X_{m'}\) have different orders, and
hence are nonisomorphic, whenever \(m\neq m'\).
\end{proof}

\smallskip 
We now return to Theorem~\ref{theorem: capacity bounds of a lexicographic product}
and present some applications that yield exact evaluations of the Shannon capacity
of lexicographic products. The following result determines this capacity for
lexicographic products of Kneser-type graphs (see Section~\ref{subsection: Kneser graphs}).

\begin{theorem}[Lexicographic products of Kneser-type graphs]
\label{theorem: lexicographic products of Kneser-type graphs}
The following statements hold.
\begin{enumerate}
\item
If \(n\geq 2k\) and \(m\geq 2\ell\), then
\begin{align}
\label{eq42: 06.08.26}
\Theta\bigl(\mathrm{KG}(n,k) \circ \mathrm{KG}(m,\ell)\bigr) = \binom{n-1}{k-1}\binom{m-1}{\ell-1}.
\end{align}

\item
If \(n\geq 2k\), \(m\geq 2\ell\), \(k\mid n\), and \(\ell\mid m\), then
\begin{align}
\label{eq43: 06.08.26}
\Theta\bigl( \, \overline{\mathrm{KG}(n,k)} \circ \overline{\mathrm{KG}(m,\ell)} \, \bigr) = \frac{nm}{k\ell}.
\end{align}

\item
If \(q\) is a prime power, \(n\geq 2k\), and \(m\geq 2\ell\), then
\begin{align}
\label{eq44: 06.08.26}
\Theta\bigl(\mathrm{KG}_q(n,k) \circ \mathrm{KG}_q(m,\ell) \bigr)
= \gbinom{n-1}{k-1}{q} \gbinom{m-1}{\ell-1}{q}.
\end{align}
\end{enumerate}
\end{theorem}

\begin{proof}
We prove the three statements separately. In each case,
the Shannon capacity and the Lov\'{a}sz theta number coincide for both
factors. Theorem~\ref{theorem: capacity bounds of a lexicographic product}
then determines the Shannon capacity of their lexicographic product.

\begin{enumerate}
\item 
For classical Kneser graphs, \cite[Theorem~13]{Lovasz79} gives 
\begin{align}
\label{eq45: 06.08.26}
\Theta\bigl(\mathrm{KG}(n,k)\bigr)
= \vartheta\bigl(\mathrm{KG}(n,k)\bigr)
= \binom{n-1}{k-1},   \qquad n \geq 2k, 
\end{align}
and, similarly,
\begin{align}
\label{eq1: 25.08.26}
\Theta\bigl(\mathrm{KG}(m,\ell)\bigr)
= \vartheta\bigl(\mathrm{KG}(m,\ell)\bigr)
= \binom{m-1}{\ell-1},   \qquad m \geq 2\ell.
\end{align}
Therefore, Theorem~\ref{theorem: capacity bounds of a lexicographic product} yields \eqref{eq42: 06.08.26}.

\smallskip 
\item 
For complements of Kneser graphs, \cite[Theorem~2.51]{LaviSason2026}
gives
\begin{align}
\label{eq46: 06.08.26}
\alpha\bigl(\overline{\mathrm{KG}(n,k)}\bigr)
= \left\lfloor\frac{n}{k}\right\rfloor,
\qquad
\vartheta\bigl(\overline{\mathrm{KG}(n,k)}\bigr)
= \frac{n}{k}.
\end{align}
Thus, if \(k\mid n\), then
\begin{align}
\label{eq2: 25.08.26}
\alpha\bigl(\overline{\mathrm{KG}(n,k)}\bigr)
= \vartheta\bigl(\overline{\mathrm{KG}(n,k)}\bigr)
= \frac{n}{k}.
\end{align}
Since
\( \alpha(F)\leq\Theta(F)\leq\vartheta(F) \)
for every finite simple graph \(F\), it follows that
\begin{align}
\label{eq3: 25.08.26}
\Theta\bigl(\overline{\mathrm{KG}(n,k)}\bigr)
= \vartheta\bigl(\overline{\mathrm{KG}(n,k)}\bigr)
= \frac{n}{k}.
\end{align}
Analogously, if \(\ell\mid m\), then
\begin{align}
\label{eq4: 25.08.26}
\Theta\bigl(\overline{\mathrm{KG}(m,\ell)}\bigr)
= \vartheta\bigl(\overline{\mathrm{KG}(m,\ell)}\bigr)
= \frac{m}{\ell}.
\end{align}
Applying Theorem~\ref{theorem: capacity bounds of a lexicographic product}
to these two graphs gives \eqref{eq43: 06.08.26}.

\smallskip 
\item 
For \(q\)-Kneser graphs,
\cite[Theorem~6.2]{LaviSason2026} gives
\begin{align}
\label{eq47: 06.08.26}
\Theta\bigl(\mathrm{KG}_q(n,k)\bigr)
= \vartheta\bigl(\mathrm{KG}_q(n,k)\bigr)
= \gbinom{n-1}{k-1}{q},
\end{align}
and, similarly, for \(\mathrm{KG}_q(m,\ell)\). A final application of
Theorem~\ref{theorem: capacity bounds of a lexicographic product}
therefore gives \eqref{eq44: 06.08.26}.
\end{enumerate}
\end{proof}

\begin{theorem}[Repeated joins]
\label{theorem: repeated joins}
For every finite simple graph \(G\) and every integer \(m \geq 1\),
\begin{align}
\label{eq: capacity of K_m[G]}
\Theta(K_m \circ G)=\Theta(G).
\end{align}
Equivalently, for every integer \(m \geq 1\),
\begin{align}
\label{eq: capacity of joins of copies of G}
\Theta\bigl(\underbrace{G\vee G\vee\cdots\vee G}_{m\text{ copies}}\bigr)=\Theta(G).
\end{align}
\end{theorem}

\begin{proof}
Let the vertices of \(K_m \circ G\) be written as pairs
$(i,v)\in[m]\times V(G)$.
Two such vertices $(i,v)$ and $(j,w)$ are adjacent if and only if 
$i\neq j$ or $\bigl( i=j\ \text{and}\ v \sim_G w \bigr)$.
For \(n \geq 1\), define the coordinate-wise projection
\begin{align}
\label{eq48: 06.08.26}
\pi \colon V\!\left((K_m \circ G)^{\boxtimes n}\right)
\longrightarrow V\!\left(G^{\boxtimes n}\right)
\end{align}
by
\begin{align}
\label{eq49: 06.08.26}
\pi\bigl((i_1,v_1),\ldots,(i_n,v_n)\bigr) = (v_1,\ldots,v_n).
\end{align}
Let \(\mathcal S\) be an independent set in
\((K_m \circ G)^{\boxtimes n}\).

First, \(\pi\) is injective on \(\mathcal S\). Indeed, suppose that two
distinct words in \(\mathcal S\) have the same image. In each coordinate
their vertices are either equal, when their first coordinates agree, or
adjacent, when their first coordinates differ. The two words would therefore
be adjacent in the strong power, contradicting the independence of
\(\mathcal S\).

Second, \(\pi(\mathcal S)\) is independent in \(G^{\boxtimes n}\). Given two
distinct words in \(\mathcal S\), their nonadjacency in the strong power gives
a coordinate in which the corresponding vertices of \(K_m \circ G\) are distinct
and nonadjacent. Nonadjacency in \(K_m \circ G\) forces their first coordinates to
be equal and their \(G\)-coordinates to be nonadjacent. Hence their projected
words are nonadjacent in \(G^{\boxtimes n}\), and therefore 
\begin{align}
\label{eq50: 06.08.26}
\alpha\!\left((K_m \circ G)^{\boxtimes n}\right)
\leq \alpha\!\left(G^{\boxtimes n}\right).
\end{align}
Conversely, choose one of the \(m\) copies of \(G\) in \(K_m\circ G\), 
and in every factor of the strong product \((K_m \circ G)^{\boxtimes n}\), 
allow only vertices belonging to that fixed copy. The resulting induced 
subgraph is isomorphic to \(G^{\boxtimes n}\). Therefore,
\begin{align}
\label{eq51: 06.08.26}
\alpha\!\left(G^{\boxtimes n}\right)
\leq \alpha\!\left((K_m \circ G)^{\boxtimes n}\right),
\end{align}
which proves the reverse inequality. Consequently, by \eqref{eq50: 06.08.26} and \eqref{eq51: 06.08.26}, 
we obtain that for all integers \( n \geq 1 \)
\begin{align}
\label{eq52: 06.08.26}
\alpha\!\left((K_m \circ G)^{\boxtimes n}\right) = \alpha\!\left(G^{\boxtimes n}\right).
\end{align}
Taking \(n\)-th roots of both sides of \eqref{eq52: 06.08.26} and letting $n \to \infty$ proves \eqref{eq: capacity of K_m[G]}.
\end{proof}

\noindent 
The next result provides another application of Theorem~\ref{theorem: capacity bounds of a lexicographic product}. 
\begin{theorem}
\label{theorem: s.c. v.t. or s.c. srg}
Let $G$ be a finite simple graph on $n$ vertices. 
\begin{enumerate}
\item If $G$ is either vertex-transitive or strongly regular, then 
\begin{align}
\label{eq53: 06.08.26}
\Theta(G \circ \overline{G}) \leq n.
\end{align}
Moreover, equality holds if $G$ is also self-complementary. 
\item If $G$ is self-complementary, then for every integer $m \geq 1$
\begin{align}
\label{eq54: 06.08.26}
\Theta(G^{\, \circ \, m}) \geq n^{\frac{m}{2}}.
\end{align}
\item If $G$ is self-complementary and either vertex-transitive or strongly regular, 
then for every integer $m \geq 1$
\begin{align}
\label{eq55: 06.08.26}
\Theta(G^{\, \circ \, m}) = n^{\frac{m}{2}}.
\end{align}
\end{enumerate}
\end{theorem}
\begin{proof}
Let $G$ be a finite simple graph on $n$ vertices.
\begin{enumerate}
\item Combining Theorem~\ref{theorem: capacity bounds of a lexicographic product} of this paper and 
Theorem~3.26 (Item~1) of \cite{Sason24} shows that, if $G$ is vertex-transitive or strongly regular, then 
\begin{align}
\label{eq56: 06.08.26}
\Theta(G \circ \overline{G}) \leq \Theta(G \boxtimes \overline{G}) = n.
\end{align}
The equality assertion for self-complementary graphs will follow from the argument in Item~2.
\item If $G$ is a self-complementary graph on $n$ vertices, then 
\begin{align}
\Theta(G) &\geq \sqrt{\alpha(G \boxtimes G)} \nonumber \\
&= \sqrt{\alpha(G \boxtimes \overline{G})} \nonumber \\
\label{eq57: 06.08.26}
&\geq \sqrt{n},
\end{align}
where the first inequality holds by definition, the equality holds since $G \cong \overline{G}$, 
and the last inequality holds since the diagonal set $\{(1,1), (2,2), \ldots, (n,n)\}$ is an independent 
set in $G \boxtimes \overline{G}$. Hence, for every positive integer $m$,
\begin{align}
\label{eq59: 06.08.26}
\Theta(G^{\, \circ \, m}) = \Theta(G)^m \geq n^{\frac{m}{2}}.
\end{align}
\smallskip
\noindent  
If $G$ is self-complementary and either vertex-transitive or strongly regular on $n$ vertices, then 
\begin{align}
\label{eq60: 06.08.26}
\Theta(G \circ \overline{G}) = \Theta(G \circ G) = \Theta(G)^2 \geq n,
\end{align}
which, by combining with item~1, gives the equality 
\begin{align}
\label{eq61: 06.08.26}
\Theta(G \circ \overline{G}) = n. 
\end{align}
This proves the equality assertion in Item~1.
\item If the graph $G$ is either self-complementary vertex-transitive or self-complementary strongly 
regular, then by Item~1, 
\begin{align}
\label{eq62: 06.08.26}
\Theta(G)^2 = \Theta(G \circ G) = \Theta(G \circ \overline{G}) \leq n,  
\end{align}
so $\Theta(G) \leq \sqrt{n}$. Combining this inequality with Item~2 gives that $\Theta(G) = \sqrt{n}$. Consequently, for every 
integer $m \geq 1$, 
\begin{align}
\label{eq63: 06.08.26}
\Theta(G^{\, \circ \, m}) = \Theta(G)^m = n^{\frac{m}{2}}.
\end{align}
\end{enumerate}
\end{proof}

\begin{example}
\label{example: Paley graphs}
{\em Let $P(n)$ be a Paley graph on $n$ vertices, where $n$ is a prime power satisfying $n \equiv 1 \pmod 4$.
This graph is self-complementary, vertex-transitive, and strongly regular (see Section~\ref{subsection: Paley graphs}). 
Consequently, by Item~3 of Theorem~\ref{theorem: s.c. v.t. or s.c. srg}, 
\begin{align}
\label{eq64: 06.08.26}
\Theta(P(n)^{\, \circ \, m}) = n^{\frac{m}{2}}.
\end{align}
In particular, we have 
\begin{align}
\label{eq65: 06.08.26}
\Theta(C_5^{\, \circ \, m}) = 5^{\frac{m}{2}},
\end{align}
where \(C_5 \cong P(5)\) denotes the pentagon (see Example~\ref{example: Paley graph P5}).}
\end{example}

\smallskip 
Combining Theorems~\ref{theorem: Shannon capacity is upper bounded by the fractional Haemers number} 
and~\ref{theorem: capacity bounds of a lexicographic product}, together with the multiplicativity 
property of the fractional Haemers number under strong products (see \cite[Theorem~3]{BukhC19}),
the following bounds on the Shannon capacity of a lexicographic product of graphs are readily obtained. 
\begin{theorem}[Bounds for the Shannon capacity]
\label{theorem2: capacity bounds of a lexicographic product}
For all finite simple graphs \(G\) and \(H\), and for every finite field $\mathbb{F}$,
\begin{align}
\Theta(G) \, \Theta(H) &\leq \min\{ \Theta(G \circ H), \Theta(H \circ G) \} \nonumber \\
& \leq \max\{ \Theta(G \circ H), \Theta(H \circ G) \} \nonumber \\
& \leq \Theta(G \boxtimes H) \nonumber \\
\label{eq2: capacity bounds of a lexicographic product}
&\leq \mathcal{H}_f(G; \mathbb{F}) \,  \mathcal{H}_f(H; \mathbb{F}).
\end{align}
\end{theorem}

\begin{example}
{\em Let \( G \) be the Schl\"afli graph, and let \(n \geq 1\)
be an integer. By Corollary~\ref{corollary: Shannon capacity of the lexicographic power of G},
since \( \alpha(G) = \Theta(G) = \vartheta(G) =3\), we have 
\[
\Theta(G^{\, \circ n}) = \Theta(G)^n = 3^n.
\]
Thus, the upper bound in
Theorem~\ref{theorem: capacity bounds of a lexicographic product}
is tight in this case. By contrast,
Theorem~\ref{theorem2: capacity bounds of a lexicographic product} gives
\[
\Theta(G^{\, \circ n})
\leq \mathcal{H}_f(G;\mathbb{R})^n
\leq \chi_f(\overline{G})^n
=\left(\frac{9}{2}\right)^n.
\]
Here, we have used the general inequality
\begin{align}
\label{eq68: 06.08.26}
\mathcal{H}_f(F;\mathbb{F})
\leq \chi_f(\overline{F}),
\end{align}
which holds for every finite simple graph \(F\) and every field
\(\mathbb{F}\), where $\chi_f(\cdot)$ denotes the fractional chromatic number 
of the graph. Moreover, since \(\overline{G}\) is vertex-transitive and \(\alpha(\overline{G})=\omega(G)=6\),
\begin{align}
\label{eq69: 06.08.26}
\chi_f(\overline{G})
=\frac{|V(G)|}{\alpha(\overline{G})}
=\frac{27}{6}
=\frac{9}{2}.
\end{align}

\noindent 
For the complement of the Schl\"afli graph \(\overline{G}\), the situation is reversed.
Theorem~\ref{theorem: capacity bounds of a lexicographic product} gives
\begin{align}
\label{eq5: 25.08.26}
\Theta(\overline{G}^{\, \circ n})
\leq \vartheta(\overline{G})^n
=9^n,
\end{align}
whereas Theorem~\ref{theorem2: capacity bounds of a lexicographic product}
gives, for every field \(\mathbb{F}\),
\begin{align}
\label{eq6: 25.08.26}
\Theta(\overline{G}^{\, \circ n})
\leq \mathcal{H}_f(\overline{G};\mathbb{F})^n
\leq 7^n.
\end{align}
Hence, for \(\overline{G}\),
Theorem~\ref{theorem2: capacity bounds of a lexicographic product}
provides a strictly sharper upper bound on the Shannon capacity of its
\(n\)-fold lexicographic power.

Consequently, the upper bounds in
Theorems~\ref{theorem: capacity bounds of a lexicographic product}
and~\ref{theorem2: capacity bounds of a lexicographic product}
are incomparable in general.}
\end{example}

\section{An open problem on graph-product capacities}
\label{section: an open problem}

In light of Theorems~\ref{theorem: capacity bounds of a lexicographic product}
and~\ref{theorem: strict supermultiplicativity of capacity},
it is natural to ask whether there exist finite simple graphs \(G\) and \(H\)
such that
\begin{equation}
\label{eq70: 06.08.26}
\max\bigl\{\Theta(G \circ H), \, \Theta(H \circ G)\bigr\} < \Theta(G \boxtimes H).
\end{equation}
Equivalently, can the Shannon capacity of the strong product of two finite
simple graphs be strictly larger than the capacities of both corresponding
ordered lexicographic products?

Propositions~\ref{proposition: independence numbers of squared lexicographic products}
and~\ref{proposition: explicit graphs satisfying strict graph complement condition}
suggest several natural candidates. Let \(S\), \(M\), and \(L\) denote,
respectively, the Schl\"{a}fli graph, the McLaughlin graph, and the second
subconstituent of the McLaughlin graph. These are vertex-transitive strongly
regular graphs with parameters
\begin{align}
S &= \operatorname{srg}(27,16,10,8),\\
M &= \operatorname{srg}(275,112,30,56),\\
L &= \operatorname{srg}(162,56,10,24).
\end{align}

If \(X\) is a vertex-transitive or strongly regular graph on \(n\) vertices,
then Theorem~3.26 of~\cite{Sason24} gives
\begin{align}
\label{eq1: 26.08.26}
\alpha(X\boxtimes\overline{X})
= \Theta(X\boxtimes\overline{X})
= \vartheta(X\boxtimes\overline{X})
= n.
\end{align}
In particular, \eqref{eq1: 26.08.26} holds for every \(X\in\{S,M,L\}\).
For each of these graphs, define
\begin{align}
b_X \coloneqq \sqrt{\lvert V(X)\rvert \, \alpha(X) \,\omega(X)}.
\end{align}
The relevant parameters and bounds are summarized in 
Table~\ref{table: parameters and capacity bounds for candidate graphs}. Here, \(U_X\) denotes 
the upper bound on \(\Theta(X) \, \Theta(\overline{X})\) established in
Proposition~\ref{proposition: explicit graphs satisfying strict graph complement condition}.
\begin{table}[ht]
\centering
\caption{\centering{Parameters and capacity bounds for the candidate graphs.}}
\label{table: parameters and capacity bounds for candidate graphs}
\renewcommand{\arraystretch}{1.35}
\begin{tabular}{|c|c|c|c|c|c|}
\hline
\(X\) & \(\lvert V(X)\rvert\) & \(\alpha(X)\) & \(\omega(X)\) & \(b_X\) & \(U_X\) \\
\hline \hline \(S\) & \(27\)  & \(3\)  & \(6\) & \(9\sqrt{6}\approx 22.045\) & \(21\)\\
\hline \(M\) & \(275\) & \(22\) & \(5\) & \(55\sqrt{10} \approx 173.925\) & \(115\)\\
\hline \(L\) & \(162\) & \(21\) & \(3\) & \(27\sqrt{14}\approx 101.025\) & \(99\)\\
\hline
\end{tabular}
\end{table}
Thus, for every \(X\in\{S,M,L\}\),
\begin{align}
\label{eq: product capacity bounds for candidates}
\Theta(X) \, \Theta(\overline{X}) \leq U_X < b_X.
\end{align}
On the other hand,
Proposition~\ref{proposition: independence numbers of squared lexicographic products}
and Theorem~\ref{theorem: capacity bounds of a lexicographic product} give
\begin{align}
b_X
&\leq \min \bigl\{\Theta(X\circ\overline{X}), \, 
                  \Theta(\overline{X}\circ X)\bigr\} \notag\\
&\leq \max\bigl\{\Theta(X\circ\overline{X}), \, 
                 \Theta(\overline{X}\circ X)\bigr\} \notag\\
&\leq \Theta(X \boxtimes \overline{X})
= \lvert V(X) \rvert.
\label{eq: candidate lexicographic and strong capacity bounds}
\end{align}
Consequently, the available bounds are
\begin{align}
9 \sqrt{6}
&\leq \min\bigl\{\Theta(S \circ \overline{S}), \, 
                 \Theta(\overline{S} \circ S)\bigr\} \notag\\
&\leq \max\bigl\{\Theta(S \circ \overline{S}), \, 
                 \Theta(\overline{S} \circ S)\bigr\}
\leq 27, 
\label{eq: Schlaefli candidate bounds} \\[0.1cm]
55 \sqrt{10}
&\leq \min\bigl\{\Theta(M \circ \overline{M}), \, 
                 \Theta(\overline{M} \circ M)\bigr\} \notag\\
&\leq \max\bigl\{\Theta(M \circ \overline{M}), \, 
                 \Theta(\overline{M} \circ M)\bigr\}
\leq 275,
\label{eq: McLaughlin candidate bounds}\\[0.1cm]
27 \sqrt{14}
&\leq \min\bigl\{\Theta(L \circ \overline{L}), \, 
                 \Theta(\overline{L} \circ L)\bigr\} \notag\\
&\leq \max\bigl\{\Theta(L \circ \overline{L}), \, 
                 \Theta(\overline{L} \circ L)\bigr\}
\leq 162.
\label{eq: second subconstituent candidate bounds}
\end{align}
It is currently unknown whether, for any \(X\in\{S,M,L\}\),
\[
\max\bigl\{\Theta(X \circ \overline{X}), \, 
           \Theta(\overline{X} \circ X) \bigr\}
<\lvert V(X) \rvert.
\]
A strict inequality in any one of these three cases would resolve 
the open problem posed in~\eqref{eq70: 06.08.26}.

It is also worth comparing these bounds with the one-shot independence
numbers. By the multiplicativity of the independence number under
lexicographic products,
\begin{equation}
\alpha(X \circ \overline{X})
= \alpha(\overline{X} \circ X)
= \alpha(X) \, \alpha(\overline{X})
= \alpha(X) \, \omega(X).
\end{equation}
Since
\( \alpha(X) \, \omega(X) < \lvert V(X)\rvert \)
for each \(X\in\{S,M,L\}\), we have
\[
\alpha(X) \, \omega(X)
< \sqrt{\lvert V(X)\rvert \,\alpha(X) \, \omega(X)} = b_X,
\]
and therefore
\begin{align}
\alpha(X \circ \overline{X})
= \alpha(\overline{X} \circ X)
< b_X
\leq \min\bigl\{\Theta(X \circ \overline{X}), \, 
                \Theta(\overline{X}\circ X)\bigr\}.
\label{eq: independence and capacity comparison for candidates}
\end{align}
Thus, for each of the three candidate graphs, the Shannon capacities of
\(X \circ \overline{X}\) and \(\overline{X} \circ X\) are both strictly
larger than their common independence number. By contrast,
\eqref{eq1: 26.08.26} shows that the Shannon capacity of
\(X \boxtimes \overline{X}\) is attained by its independence number.

More generally,
Theorem~\ref{theorem: strict supermultiplicativity of capacity} provides 
three countably infinite families of candidate graph--complement pairs. 
For \(X \in \{S,M,L\}\) and an integer \(m \geq 1\), let
\(X_m \coloneqq X^{\,\circ m}\), as in
Theorem~\ref{theorem: strict supermultiplicativity of capacity}.
Since lexicographic powers of vertex-transitive graphs are
vertex-transitive, it follows that \( X_m \) is vertex-transitive, and 
\begin{align}
\label{eq: strong capacities of candidate powers}
\Theta(X_m \boxtimes \overline{X_m})
= \lvert V(X_m)\rvert
= \lvert V(X)\rvert^m.
\end{align}
At the same time,
Theorem~\ref{theorem: strict supermultiplicativity of capacity} gives
\begin{align}
\min\bigl\{\Theta(X_m \circ \overline{X_m}), \, 
           \Theta(\overline{X_m} \circ X_m)\bigr\}
> \Theta(X_m) \, \Theta(\overline{X_m}).
\end{align}
It remains open whether, for some graph \(X \in \{S,M,L\}\) and
some integer \(m \geq 1\),
\begin{align}
\max\bigl\{ \Theta(X_m \circ \overline{X_m}), \, 
\Theta(\overline{X_m} \circ X_m) \bigr\}
&< \Theta(X_m \boxtimes \overline{X_m}).
\end{align}
A positive answer would resolve the open problem posed in
\eqref{eq70: 06.08.26}.

\end{document}